\RequirePackage{snapshot}
\newif\ifproofs\proofsfalse\ifproofs\RequirePackage[displaymath,mathlines]{lineno}\fi

\newif\ifsubsections
\subsectionsfalse
\documentclass[]{pcmi}

\usepackage[sc]{mathpazo}          
\usepackage{eulervm}               
\usepackage[mathscr]{eucal}     
\usepackage[scaled=0.86]{berasans} 
\usepackage[scaled=1]{inconsolata} 
\usepackage[T1]{fontenc}
\usepackage{overpic}
\usepackage{float}
 
\usepackage[%
	protrusion=true,
	expansion=false,
	auto=false
	]{microtype}

\usepackage{xcolor}
\usepackage{graphicx}
\graphicspath{{./figures/}}

\long\def\commentout#1{}
\ifdraft
	\definecolor{linkred}{rgb}{0.7,0.2,0.2}
	\definecolor{linkblue}{rgb}{0,0.2,0.6}
\else
	\definecolor{linkred}{rgb}{0.0,0.0,0.0}
	\definecolor{linkblue}{rgb}{0,0.0,0.0}
\fi

\usepackage{tikz}
\usetikzlibrary{calc,arrows,shapes,positioning}
\usetikzlibrary{decorations.markings}
\tikzstyle arrowstyle=[scale=1]
\tikzstyle directed=[postaction={decorate,decoration={markings,
    mark=at position .65 with {\arrow[arrowstyle]{stealth}}}}]
\tikzstyle reverse directed=[postaction={decorate,decoration={markings,
    mark=at position .65 with {\arrowreversed[arrowstyle]{stealth};}}}]
\usepackage{float}
\usepackage[skip=6pt]{caption}

\PassOptionsToPackage{hyphens}{url}
\usepackage[
    setpagesize=false,
    pagebackref,
	plainpages=false,pdfpagelabels,pageanchor=true,
    pdfstartview={FitH 1000},
    bookmarksnumbered=false,
    linktoc=all,
    colorlinks=true,
    anchorcolor=black,
    menucolor=black,
    runcolor=black,
    filecolor=black,
    linkcolor=linkblue,
	citecolor=linkblue,
	urlcolor=linkred,
]{hyperref}
\usepackage[backrefs,msc-links,nobysame]{amsrefs}

\customizeamsrefs

\theoremstyle{plain}
\newtheorem{theorem}[equation]{Theorem}
\newtheorem{corol}[equation]{Corollary}
\newtheorem{prop}[equation]{Proposition}
\newtheorem{define}[equation]{Definition}

\theoremstyle{definition}

\newtheorem{exer}[equation]{Exercise}
\newtheorem{rem}[equation]{Remark}
\newtheorem*{remark}{Remark}
\newtheorem*{remarks}{Remarks}
\newtheorem{exam}[equation]{Example}
\newtheorem*{ques}{Question}

\DeclareRobustCommand{\ddbar}{%
  {\mathord{\text{\lower0.05ex\hbox{$\mathchar'26$}{$\mkern-11mu d$}}}}%
}
\newcommand{\dts}{\ddbar s}

\newcommand{\dtx}{\ddbar x}
\newcommand{\dtz}{\ddbar z}

\newcommand{\lf}{\left}
\newcommand{\rt}{\right}

\newcommand{\tit}{\textit}
\newcommand{\tbf}{\textbf}

\newcommand{\ovl}{\overline}
\newcommand{\tconst}{\text{const.}}
\newcommand{\ndconst}{\text{const }}

\newcommand{\sgn}{\text{sgn}}
\newcommand{\ind}{\text{ind}}
\newcommand{\noi}{}

\newcommand{\beq}{\begin{equation}}
\newcommand{\eeq}{\end{equation}}
\newcommand{\beqs}{\begin{equation*}}
\newcommand{\eeqs}{\end{equation*}}

\newcommand{\RR}{\mathbb{R}}
\newcommand{\TT}{\mathbb{T}}
\newcommand{\CC}{\mathbb{C}}
\newcommand{\mL}{\mathcal{L}}
\newcommand{\mP}{\mathcal{P}}
\newcommand{\mR}{\mathcal{R}}
\newcommand{\msA}{\mathscr{A}}
\newcommand{\msC}{\mathscr{C}}
\newcommand{\msF}{\mathscr{F}}
\newcommand{\mS}{\mathcal{S}}
\newcommand{\chf}{\check{f}}

\newcommand{\hC}{\hat{C}}

\newcommand{\hf}{\hat{f}}
\newcommand{\hh}{\hat{h}}
\newcommand{\hm}{\hat{m}}

\newcommand{\tc}{\tilde{c}}
\newcommand{\ttc}{\tilde{\tc}}

\newcommand{\hz}{\hat{z}}

\newcommand{\al}{\alpha}
\newcommand{\ga}{\gamma}
\newcommand{\Ga}{\Gamma}
\newcommand{\de}{\delta}

\newcommand{\ep}{\epsilon}
\newcommand{\om}{\omega}

\newcommand{\sg}{\sigma}
\newcommand{\Sg}{\Sigma}
\newcommand{\la}{\lambda}

\newcommand{\p}{\partial}

\newcommand{\half}{\frac{1}{2}}

\newcommand{\arcThroughThreePoints}[4][]{
\coordinate (middle1) at ($(#2)!.5!(#3)$);
\coordinate (middle2) at ($(#3)!.5!(#4)$);
\coordinate (aux1) at ($(middle1)!1!90:(#3)$);
\coordinate (aux2) at ($(middle2)!1!90:(#4)$);
\coordinate (center) at ($(intersection of middle1--aux1 and middle2--aux2)$);
\draw[thick,#1] 
 let \p1=($(#2)-(center)$),
      \p2=($(#4)-(center)$),
      \n0={veclen(\p1)},       
      \n1={atan2(\y1,\x1)}, 
      \n2={atan2(\y2,\x2)},
      \n3={\n2>\n1?\n2:\n2+360}
    in (#2) arc(\n1:\n3:\n0);
}

\begin{document}

%
%
%
%
%
%

\title{Riemann--Hilbert Problems}
\author{Percy Deift}
\address{Department of Mathematics, Courant Institute of Mathematical Sciences, New York University}
\email{
deift@cims.nyu.edu}

\begin{abstract}%
These lectures introduce the method of nonlinear steepest descent for Riemann-Hilbert problems. This method finds use in studying asymptotics associated to a variety of special functions such as the Painlev\'{e} equations and orthogonal polynomials,
in solving the inverse scattering problem for certain integrable systems, and in proving universality for certain classes of random matrix ensembles. These lectures highlight a few such applications.
\end{abstract}

\maketitle
\tableofcontents
\thispagestyle{empty}

\section*{Lecture 1}
\setcounter{section}{1}\addcontentsline{toc}{section}{Lecture 1}

These four lectures are an abridged version of 14 lectures that I gave at
the Courant Institute on RHPs in 2015.  These 14 lectures are freely available
on the AMS website AMS Open Notes.

Basic references for RHPs are \cite{ClanceyGohberg,Litvinchuk1987,DeiftOrthogonalPolynomials}.  Basic references for complex function theory are \cite{Duren,Garnett2007,Goluzin1969}.
Many more specific references will be given as the course proceeds.

Special functions are important because they provide \tit{explicitly
solvable models} for a vast array of phenomena in mathematics and physics.
By ``special functions'' I mean Bessel functions, Airy functions, Legendre
functions, and so on.  If you have not yet met up with these functions, be
assured, sooner or later, you surely will.

It works like this.  Consider the \tit{Airy equation} (see, e.g. \cite{AbramowitzStegun,DLMF})
\beq
y''(x) = x y(x),\quad  -\infty < x < \infty.
\label{8}
\eeq
Seek a solution of \eqref{8} in the form
\beqs
y(x) = \int_\Sg e^{xs} \: f(s)\, ds
\eeqs
for some functions $f(x)$ and some contours $\Sg$ in the complex plane
$\CC$.  We have
\begin{align*}
y'' (x) &= \int_\Sg s^2\: e^{xs}\: f(s)\, ds\\[-3mm]
\intertext{and}
x\, y(x) &= \int_\Sg \lf(\frac{d}{ds} \;e^{xs} \rt) f(s)\, ds\\[2mm]
&= -\int e^{xs} \, f'(s)\, ds
\end{align*}
provided we can drop the boundary terms.  In order to solve \eqref{8} we need to
have
\begin{align*}
-f'(s) &= s^2\, f\\
\intertext{and so}
f(s) &= \tconst \; e^{-\frac{1}{3}\;s^3}\ .
\end{align*}
Thus
\[
y(x) = \tconst \int_\Sg e^{xs - \frac{1}{3}\; s^3}\, ds
\]
provides a solution of the Airy equation.

The particular choice
\[
\tconst = \frac{1}{2\pi i}
\]
and $\Sigma$ in Figure~\ref{fig:sigma}
is known as \tit{Airy's integral} $Ai(x)$
\beq
Ai(x) = \frac{1}{2\pi i} \int_\Sg e^{xz- \frac{1}{3}\;z^3}\: dz\,.
\label{10}
\eeq

\begin{figure}[H]
\centering
\begin{tikzpicture}[scale=6]
\draw[line width = 1, directed]  (0, -.3)to [out=30,in= 270] (.28,0) to [out=90, in = -30] (0, .3);
\node[above right] at (0,.3){$\infty e^{i2\pi/3}$};
\node[below right] at (0,-.3){$\infty e^{-i2\pi/3}$};
\node[left] at (-0, 0) {\scalebox{2}{$\Sigma = $}};
\end{tikzpicture}
\caption{$\Sigma$ for Airy's integral. \label{fig:sigma}} 
\end{figure}
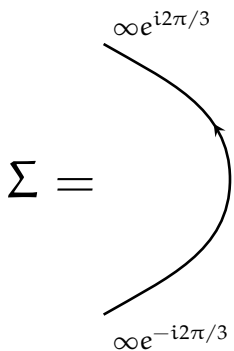

Other contours provide other, independent solutions of Airy's equation, such
as $Bi(x)$ (see \cite{AbramowitzStegun}).  Now the \tit{basic fact of the matter} is
that the integral representation \eqref{10} for $Ai(x)$ enables us, using the
classical method of \tit{stationary phase/steepest descent}, to compute
the asymptotics of $Ai(x)$ as $x\to + \infty$ and $-\infty$ with
\tit{any desired accuracy}.  We find, in particular \cite[p.~448]{AbramowitzStegun}, that for $ \zeta = \frac23 x^{3/2}$
\begin{equation}\label{11}
Ai(x) \sim \frac{1}{2\sqrt{\pi}}\; x^{-\frac{1}{4}} \;e^{-\zeta}
\sum^\infty_{k=0} (-1)^k\, c_k\, \zeta^{-k} 
\end{equation}
as $x\to +\infty$, where
\begin{align*}
c_0 &= 1,\\  c_k &= \frac{\Ga \lf(3k+ \frac{1}{2}\rt)}{{54}^k\;
k!\; \Ga\lf(k+\half\rt)}
= \frac{(2k+1) (2k+3) \dots (6k-1)}{(216)^k\; k!}\ , \qquad k \ge 1\,.
\end{align*}
and that
\begin{equation} \label{12} 
\begin{split}
	Ai(-x) \sim\; \frac{1}{\sqrt{\pi}} \;\; x^{-1/4}
	\biggl(&\sin
	\lf( \zeta + \frac{\pi}{4} \rt)
	\sum^\infty_0 \, (-1)^k \: c_{2k} \: \zeta^{-2k} \\
	-&\cos \lf(\zeta + \frac{\pi}{4}\rt) \sum^\infty_0 (-1)^k\:
	c_{2k+1} \: \zeta^{-2k-1}\biggr), 
\end{split}
\end{equation}
as $x\to +\infty$.

Such results for solutions of general 2$^{nd}$ order equations are very rare.
Formulae \eqref{11} and \eqref{12} solve the fundamental \tit{connection problem}
or \tit{scattering problem} for solutions of the Airy equation.  Thus, if
we know that a solution $y(x)$ of the Airy equation behaves like
\[
y(x) = \frac{1}{2\sqrt{\pi}} \;x^{-1/4} \; e^{-\zeta} \lf(1-\frac{c_1}{\zeta}
+ \dots \rt)
\]
as $x\to + \infty$, then we know \tit{precisely} how it behaves as
$x\to -\infty$, and vice versa, by \eqref{11} \eqref{12}, see Figure~\ref{fig:airy}.

\begin{figure}[H]

  \centering\begin{overpic}[width=.65\linewidth]{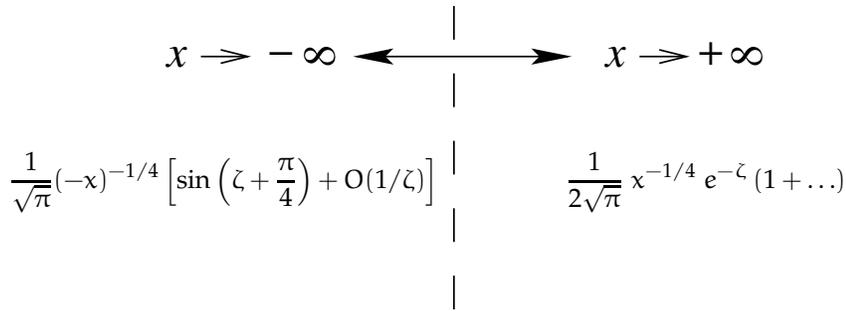}
    \put(-25,20){$\displaystyle
\frac{1}{\sqrt{\pi}}(-x)^{-1/4} \lf[ \sin \lf( \zeta +  \frac{\pi}{4}\rt) + O(1/\zeta) \rt]$}
\put(65,20){$\displaystyle \frac{1}{2\sqrt{\pi}} \; x^{-1/4} \;e^{-\zeta} \lf(1+ \dots \rt)$}
\end{overpic}
\caption{\label{fig:airy} Asymptotics for Airy's integral.}
\end{figure}

\begin{exer}
Use the classical steepest-descent method to verify \eqref{11} and \eqref{12}.
There are similar precise results for \tit{all} the classical special
functions. The diligent student should regard Abramowitz \& Stegun \cite{AbramowitzStegun} as an
exercise book for the steepest descent method --- verify all the
asymptotic formulae!

Now in recent years it has become clear that a new and extremely broad class
of problems in mathematics, engineering and physics is described by a
\tit{new} class of special functions, the so-called \tit{Painlev\'e
functions}.  There are six Painlev\'e equations and we will say more
about them later on.  Whereas the classical special functions, such as
Airy functions, Bessel functions, etc. typically arise in linear
(or linearized problems) such as acoustics or electromagnetism, the
Painlev\'e equations arise in nonlinear problems, and they are now
recognized as forming the core of modern special function theory.  Here are
some examples of how Painlev\'e equations arise:
\end{exer}

\begin{exam}
Consider solutions of the modified Korteweg--de Vries equation (MKdV)
\begin{equation}\label{13}
\begin{split}
	&u_t - 6u^2\,u_x + u_{xxx} = 0,\; - \infty< x < \infty, \quad t>0, \\
	&u(x,0) = u_0 (x) \to 0 \quad\text{as}\quad |x| \to \infty. 
\end{split}
\end{equation}
Then \cite{deiftzhoumkdv} as $t\to \infty$, in the region $|x| \le c \; t^{1/3}, \;\;c<\infty$,
\beq
u(x,t) = \frac{1}{(3t)^{1/3}} \quad p\lf(\frac{x}{(3t)^{1/3}} \rt)
+ O\lf(\frac{1}{t^{2/3}}\rt)
\label{14}
\eeq
where $p(s)$ is a particular solution of the \tit{Painlev\'e II (PII)}
equation
\beqs
p''(s) = s \;p(s) + 2 \;p^3(s).
\eeqs
\end{exam}

\begin{exam} Let $\pi := \lf(\pi_1\, \pi_2 \dots \pi_N\rt) \in S_N$ be a permutation of the numbers
$1, 2, \dots, N$.  We say that $\pi_{i_1}, \pi_{i_2},
\dots \pi_{i_k}$ is an \tit{increasing subsequence} of $\pi$ of \tit{length k}
if
\[
i_1 < i_2 < \dots < i_k
\]
and
\[
\pi_{i_1} < \pi_{i_2} < \dots < \pi_{i_k}.
\]
Thus if $N=6$ and $\pi=(413265)$, then 125 and 136 are increasing subsequences
of $\pi$ of length 3. Let $\ell_N(\pi)$ denote the length of a longest
increasing subsequence of $\pi$, e.g., for $N=6$ and $\pi$ as above,
$\ell_6(\pi)=3$, which is the length of the longest increasing subsequences
125 and 136.

Now equip $S_N$ with uniform measure. Thus
\beqs
\text{Prob } \lf(\ell_N \le n\rt) = \frac{\# \; \lf\{ \pi \in S_N:
\ell_N(\pi) \le n\rt\}}{N!}. 
\eeqs
\end{exam}
\begin{ques}
How does $\ell_N$ behave statistically as $N,n\to \infty$?
\end{ques}

\begin{theorem}[\cite{Baik1999b}] \label{t:long}

Center and scale $\ell_N$ as follows:
\[
\ell_N \to X_N = \frac{ \ell_N - 2\sqrt{N}}{N^{1/6}}
\]
then
\beqs
\lim_{N\to \infty} \textup{Prob} \lf(X_N \le x\rt) =
e^{-\int^\infty_x \; (s-x)\, u^2(s)\, ds}
\eeqs
where $u(s)$ is the (unique) solution of Painlev\'e II (the so-called
Hastings-McLeod solution) normalized such that
\beqs
u(s) \sim Ai(s) \quad\text{as} \quad s\to + \infty.
\eeqs
\end{theorem}
The distribution on the right in Theorem \ref{t:long} is the famous
Tracy-Widom distribution for the largest eigenvalue of a GUE matrix in
the edge scaling limit.  Theorem 1 is one of a very large number of
probabilistic problems  in combinatorics and related areas, whose solution
is expressed in terms of \tit{Random Matrix Theory (RMT)} via Painlev\'e
functions (see, e.g., \cite{Baik2017}).


The \tit{key} question is the following:  Can we describe the solutions of the Painlev\'e equations as precisely
as we can describe the solutions of the classical special functions such
as Airy, Bessel, $\dots$ ? In particular, can we describe the solutions of
the Painlev\'e equations asymptotically with arbitrary precision and solve
the connection/scattering problem as in \eqref{11} and \eqref{12} for the Airy
equation (or any other of the classical special functions):
\begin{equation*}
	\text{known behavior as }\; x\to + \infty
	\qquad \Rightarrow\qquad
	\text{known behavior as }\; x\to - \infty
\end{equation*}
and vice versa.

As we have indicated, at the technical level, connection formulae such
as \eqref{11} and \eqref{12} can be obtained because of the existence of an integral
representation such as \eqref{10} for the solution.  Once we have such a
representation the asymptotic behavior is obtained by applying the
(classical) steepest descent method to the integral.  There are,  however,
no known integral representations for solutions of the Painlev\'e equations
and we are led to the following questions:

\noi \tbf{Question 1:}  Is there an analog of an integral representation
for solutions of the Painlev\'e equations?

\noi \tbf{Question 2:}  Is there an analog of the classical steepest
descent method which will enable us to extract precise asymptotic information
about solutions of the Painlev\'e equations from this analog representation?

The answer to both questions is \tit{yes}: In place of an integral
representation such as \eqref{10}, we have a \tit{Riemann--Hilbert Problem (RHP)},
and in place of the classical steepest descent method we have the
\tit{nonlinear (or non-commutative) steepest descent method} for RHPs (introduced by P.~Deift
and X.~Zhou \cite{deiftzhoumkdv}).

So what is a RHP?  Let $\Sg$ be an oriented contour in the plane, see Figure~\ref{fig:2}.
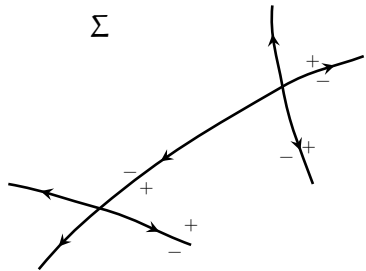
\begin{figure}[H]
\centering
\begin{tikzpicture}[scale =1.35]
\draw[line width = 1, directed] (0,0) to [out=220, in=50](-.6,-.6);
\draw[line width = 1, directed] (0,0) to [out=165, in=-10](-.9,.24);
\draw[line width = 1, directed] (0,0) to [out=-15, in=160](.9,-.36);
\draw[line width = 1, directed] (1.8,1.2) to [out=210, in = 40] (0,0);
\draw[line width = 1, directed] (1.8,1.2) to [out=100, in = 265] (1.7,2);
\draw[line width = 1, directed] (1.8,1.2) to [out=280, in = 110] (2.1,.24);
\draw[line width = 1, directed] (1.8,1.2) to [out=30, in = 200] (2.6,1.5);
\node[above] at (.3, .2) {\scalebox{.7}{$-$}};
\node[right] at (.3, .2) {\scalebox{.7}{$+$}};
\node[above] at (.9, -.32) {\scalebox{.7}{$+$}};
\node[below left] at (.9, -.28) {\scalebox{.7}{$-$}};
\node[above] at (2.1, 1.3) {\scalebox{.7}{$+$}};
\node[below] at (2.2, 1.4) {\scalebox{.7}{$-$}};
\node[right] at (1.9, .6) {\scalebox{.7}{$+$}};
\node[left] at (2, .5) {\scalebox{.7}{$-$}};
\node at (0, 1.8) {\scalebox{1.2}{$\Sigma$}};
\end{tikzpicture}
\caption{\label{fig:2} An oriented contour in the plane.} 
\end{figure}
\noi By convention, if we move along an arc in $\Sg$ in the direction of the
orientation, the $(\pm)$-sides lie on the left (resp.\ right).  Let $v:
\Sg \to \textup{GL} (k, \CC)$, the \tit{jump matrix}, be an invertible
$k \times k$ matrix function defined on $\Sg$ with
\beqs
v, v^{-1} \in L^\infty (\Sg). 
  \eeqs
We say that an $n\times k$ matrix function $m(z)$ is a \tit{solution}
of the RHP $(\Sg, v)$ if\\
\begin{minipage}{.6\textwidth}
\begin{align}
&m(z) \; \text{ is analytic in } \CC/\Sg, \notag\\
&m_+(z) = m_-(z) v(z), z \in \Sg, \notag \\
&\text{where } m_\pm (z) = \lim_{z' \to z_\pm }  m(z').\notag
\end{align}
\end{minipage}
\begin{minipage}{0.4\textwidth}
\begin{tikzpicture}[scale=1]
\fill (1,0.6) circle (.05cm);
\draw[line width=1, directed] (0,0) to (2,1.2);
\draw[dashed] (1,0) to (1,.6);
\draw[dashed] (.4,.9) to (1,.6);
\node[above left] at (1.8, .9) {\scalebox{.7}{$+$}};
\node[below] at (1.8, 1) {\scalebox{.7}{$-$}};
\node[below right] at (1,.6) {$z$};
\node[below right] at (0.8,0){$z' \to z^{-}$};
\node[above left] at (.4, .9) {$z' \to z^{+}$};
\end{tikzpicture} 
\end{minipage}
If, in addition, $n=k$ and
\beqs
m(z) \to I_k \quad\text{as }\quad z\to \infty, 
\eeqs
we say that $m(z)$ solves the \tit{normalized} RHP $(\Sg, v)$.

RHPs involve a lot of technical issues.  In particular
\begin{itemize}
\item How smooth should $\Sg$ be?
\item What measure theory/function spaces are suitable for RHPs?
\item What happens at points of self intersection (see Figure~\ref{fig:3})?
\end{itemize}

\begin{figure}[H] 
\centering
\begin{tikzpicture}
\fill (0,0) circle (.066cm);
\draw[line width = 1, directed] (0,0) to [out=30, in = 190](1.5, .6);
\draw[line width = 1, directed] (1.2,-.6) to [out=150, in = -20](0,0);
\draw[line width = 1, directed] (0,0) to [out=165, in = 0](-1.5, .3);
\draw[line width = 1, directed] (0,0) to [out=210, in = 45](-1.1, -.8);
\end{tikzpicture}
\caption{\label{fig:3} A point of self intersection.} 
\end{figure}
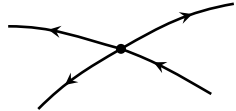

\begin{itemize}
\item In what sense are the limits $m_\pm(z)$ achieved?
\item In the case $n=k$, in what sense is the limit $m(z) \to I_k$
achieved?
\item Does an $n\times k$ solution exist?
\item In the normalized case, is the solution unique?
\end{itemize}
And most importantly
\begin{itemize}
\item at the analytical level, what kind of problem is a RHP?  As we will
see, the problem reduces to the analysis of singular integral equations on
$\Sg$.
\end{itemize}

There is not enough time in these 4 lectures to address all these issues
systematically.  Rather we will address specific issues as they arise.

As an example of how things work, we now show how PII is related to
a RHP
(see, e.g. \cite{FokasPainleve}). Let $\Sg$ denote the union of six rays
\[
\Sg_k = e^{i(k-1)\, \pi/3}\; \rho,\quad \rho>0, \quad 1\le  k \le 6
\]
oriented outwards.
Let $p, q, r$ be complex numbers satisfying the relation
\beq
p + q+r + pqr =0.
\label{25}
\eeq
Let $v(z), \;\;z\in \Sg$, be constant on each ray as indicated in Figure~\ref{fig:4}
and for \tit{fixed} $x \in \CC$ set
\beqs
v_x(z) = \begin{pmatrix}
e^{-i\theta} &0 \\
0 & e^{i\theta}
\end{pmatrix} \; v(z) \;
\begin{pmatrix}
e^{i\theta} &0 \\
0 & e^{-i\theta}
\end{pmatrix}, \qquad z\in \Sg
\eeqs
where
\beqs
\theta = \theta_x(z) = \frac{4}{3} z^3 + xz.
\eeqs
Thus for $z\in \Sg_3$
\[
v_x(z) = \begin{pmatrix}
1 & r\, e^{-2i\theta}\\
0 &1
\end{pmatrix}
\]
and so on.

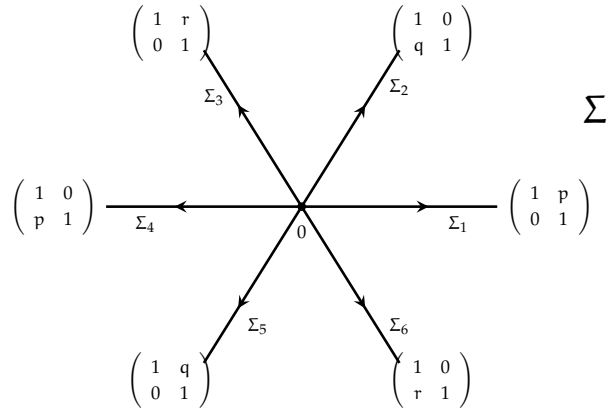
\begin{figure}[H]
\centering
\begin{tikzpicture}[scale=1.3]
\node at (3, 1) {\scalebox{1.5}{$\Sigma$}};
\fill (0,0) circle (.044cm);
\node [below] at (0,-.1) {\scalebox{.7}{$0$}};
\draw[line width=1, directed] (0,0) to (2,0);
\node[below] at (1.6, 0){\scalebox{.7}{$\Sigma_1$}};
\draw[line width=1, directed] (0,0) to (-2,0);
\node[below] at (-1.6, 0){\scalebox{.7}{$\Sigma_4$}};
\draw[line width=1, directed] (0,0) to (1,1.6);
\node[below right] at (.8, 1.4) {\scalebox{.7}{$\Sigma_2$}};
\draw[line width=1, directed] (0,0) to (1,-1.6);
\node[above right] at (.8, -1.4) {\scalebox{.7}{$\Sigma_6$}};
\draw[line width=1, directed] (0,0) to (-1,1.6);
\node[below left] at (-.7, 1.3) {\scalebox{.7}{$\Sigma_3$}};
\draw[line width=1, directed] (0,0) to (-1,-1.6);
\node[above right] at (-.65, -1.4) {\scalebox{.7}{$\Sigma_5$}};
\node[right] at (2,0) {\scalebox{.7}{$\left(\begin{array}{cc}
1 & p \\
0 & 1
\end{array}\right)$}};
\node[left] at (-2,0) {\scalebox{.7}{$\left(\begin{array}{cc}
1 & 0 \\
p & 1
\end{array}\right)$}};
\node[above right] at (.8, 1.4) {\scalebox{.7}{$\left(\begin{array}{cc}
1 & 0 \\
q & 1
\end{array}\right)$}};
\node[below right] at (.8, -1.4) {\scalebox{.7}{$\left(\begin{array}{cc}
1 & 0 \\
r & 1
\end{array}\right)$}};
\node[above left] at (-.8, 1.4) {\scalebox{.7}{$\left(\begin{array}{cc}
1 & r \\
0 & 1
\end{array}\right)$}};
\node[below left] at (-.8, -1.4) {\scalebox{.7}{$\left(\begin{array}{cc}
1 & q \\
0 & 1
\end{array}\right)$}};
\end{tikzpicture}
\caption{\label{fig:4}  Six rays oriented outwards. } 
\end{figure}

For \tit{fixed} $x$, let $m_x(z)$ be the $2 \times 2$ matrix solution of the
normalized RHP $(\Sg, v_x)$.  Then
\beqs
u(x) = 2i(m_1(x))_{12}
\eeqs
is a solution of the PII equation where
\beqs
m_x(z) = I + \frac{m_1 (x)}{z} + O\lf(\frac{1}{z^2}\rt)
\eeqs
as $z\to\infty$.  (This result is due to Jimbo and Miwa \cite{Jimbo}, and
independently to Flaschka and Newell \cite{Flaschka}.)
The asymptotic behavior of $u(x)$ as $x\to \infty$ is then obtained from the
RHP $(\Sg, v_x)$ by the nonlinear steepest descent method.

In the classical steepest descent method for integrals such as \eqref{10} above, the contour $\Sigma$ is deformed so that the integral passes through a stationary phase point where the integrand is maximal and the main contribution to the integral then comes from a neighborhood of this point.   The nonlinear (or non-commutative) steepest descent method for RHPs involves the same basic ideas as in the classical scalar case in that one deforms the RHP, $\Sigma \to \Sigma'$, in such a way that the exponential terms (see e.g. $e^{2 i \theta}$ above) in the RHP have maximal modulus at points of the deformed contour $\Sigma'$.  The situation is far more complicated than the scalar integral case, however, as the problem involves matrices that do not commute.  In addition, terms of the form $e^{-2 i \theta}$ also appear in the problem and must be separated algebraically from terms involving $e^{2 i \theta}$, so that in the end the terms involving $e^{2 i \theta}$ and $e^{-2i \theta}$ both have maximal modulus along $\Sigma'$ (see \cite{deiftzhoumkdv,deiftzhounls,Deift1995a}).  A~simple example of the nonlinear steepest descent method is given at the end of Lecture 4.

One finds, in particular, (\cite{Deift1995a}, and also \cite{Its1994,FokasPainleve}) the following:

Let $-1 < q <1, \; p=-q, \; r=0$.  Then as $x\to -\infty$,
\begin{equation}\label{30} 
u(x) = \frac{\sqrt{2\nu}}{(-x)^{1/4}} \;\cos \lf(\frac{2}{3} \lf(-x\rt)^{3/2}
-\frac{3}{2} \; \nu \;\log (-x) + \phi\rt)
+ O \lf(\frac{\log (-x)}{(-x)^{5/4}}\rt) 
\end{equation}
where
\beq
\nu= \nu(q) = -\frac{1}{2\pi} \;\log \lf(1-q^2\rt)\label{31}
\eeq
and
\beq
\phi = -3 \nu \; \log 2 + \arg \Ga(i\nu) + \frac{\pi}{2}\; \sgn (q)
- \frac{\pi}{4}.
\label{32}
\eeq
As $x\to + \infty$
\beq
u(x) = q \;Ai(x) + O\lf(\frac{e^{-4/3 \;x^{3/2}}}{x^{1/4}} \rt).
\label{33}
\eeq

These asymptotics should be compared with \eqref{11}, \eqref{12} for the Airy function.
Note from \eqref{11} that as $x\to + \infty$
\[
Ai(x) \sim x^{-1/4} e^{-2/3\; x^{3/2}}.
\]
Also observe that PII
\[
u''(x) = x\; u(x) + 2\, u^3(x)
\]
is a clearly a nonlinearization of the Airy equation
\[
u''(x) = x\; u(x)
\]
and so we expect similar solutions when the nonlinear
term $2\,u^3(x)$ is small.

Also note that \eqref{30} and \eqref{31} solve the connection problem for PII.  If
we know the behavior of the solutions $u(x)$ of PII as $x\to +\infty$,
then we certainly know $q$ from \eqref{33}.  But then we know $\nu=\nu(q)$
and $\phi=\phi(q)$ in \eqref{31} and \eqref{32} and hence we know the asymptotics
of $u(x)$ as $x\to -\infty$ from \eqref{30}.  Conversely, if we know the
asymptotics of $u(x)$ as $x\to -\infty$, we certainly know $\nu >0$ from
\eqref{30} and hence we know $q^2$ from \eqref{31}, $q^2=1- e^{-2\pi\,\nu}$.  But
then again from \eqref{30}, we know $\phi$, and hence $\sgn(q)$ from \eqref{32}.  Thus
we know $q$, and hence the asymptotics of the solution $u(x)$ as $x\to +\infty$
from \eqref{33}.  Finally note the similarity of the multiplier
\beq
e^{x\, z- \frac{1}{3}\;z^3}
\label{34}
\eeq
for the Airy equation with the multiplier
\beq
e^{i\theta}= e^{i\lf(x\;z + \frac{4}{3}\;z^3\rt)}
\label{35}
\eeq
in the RHP for PII.  Setting $z\to i\;z$ in \eqref{34}
\[
e^{x\;z - \frac{1}{3} \;z^3} \to e^{i\lf(x\;z + \frac{1}{3}\;z^3\rt)}
\]
which agrees with \eqref{35} up to appropriate scalings.

Also note from \eqref{25} that PII is parameterized by parameters lying on a
2-dim variety:  this corresponds to the fact that PII is second order.

The fortunate and remarkable fact is that the class of problems in physics, mathematics, and engineering expressible in terms of a RHP is very broad and growing.  Here
is one more, with more to come!

The RHP for the MKdV equation \eqref{13} is as follows (see e.g., \cite{deiftzhoumkdv}):  Let $\Sg=\RR$, oriented from $-\infty$ to $+\infty$.  For
fixed $x, \; t \in \RR$ let
\beq
v_{x,t}(z) = \begin{pmatrix}
1-|r(z)|^2 &-\ovl{r(z)} \;e^{-2i\tau}\\
r(z) \;e^{2i\tau}  &1
\end{pmatrix}, \qquad z\in \RR
\label{36}
\eeq
where $\tau=\tau_{x,t}(z) = xz +4 t z^3$ and $r=r(z)$ is a given function
in $L^\infty(\RR) \cap L^2(\RR)$ with
\beqs
\|r\|_\infty <1
\eeqs
and
\beqs
r(z) =- \ovl{r(-z)}, \quad z\in \RR. 
\eeqs
There is a bijection from the initial data $u(x,t=0)=u_0(x)$ for
MKdV onto such functions $r(z)$ --- see later.  The function $r(z)$ is
called the \tit{reflection coefficient} for $u_0$, see (\ref{eq:reflect}).

Let $m=m_{x,t}(z)$ be the solution of the normalized RHP
$\lf(\Sg, v_{x,t}\rt)$.  Then
\beq
u(x,t) = 2\lf(m_1 (x,t)\rt)_{12} \ ,
\label{39}
\eeq
is the solution of MKdV with initial condition $u(x,t=0)=u_0(x)$ corresponding
to $r(z)$.
Here
\[
m_{x,t}(z) = I + \frac{m_1(x,t)}{z} + O\lf(\frac{1}{z^2}\rt)
\]
as $z\to \infty$.

\begin{figure}[H]
\centering
\begin{tikzpicture}[scale=0.6]
\fill (0,0) circle (.06cm);
\draw[line width =1, directed] (-2, 0) to (0,0);
\draw[line width =1, directed] (0, 0) to (2,0);
\node[below] at (0,0){$\Sigma = \mathbb{R}$};

\node at (4,0) {\scalebox{1.5}{$\to$}};

\draw[line width =1, directed] (6, 0) to (8,0);
\draw[line width =1, directed] (8, 0) to (10,0);
\draw[line width =1, directed] (7, -1.5) to (8,0);
\draw[line width =1, directed] (8, 0) to (9,1.5);

\node at (4,-4) {\scalebox{1.5}{$\to$}};

\draw[line width =1, directed] (6, -4) to (8,-4);
\draw[line width =1, directed] (8, -4) to (10,-4);
\draw[line width =1, directed] (7, -5.5) to (8,-4);
\draw[line width =1, directed] (8, -4) to (9,-2.5);
\draw[line width =1, dashed] (7, -2.5) to (9,-5.5);
 
\node[right] at (9, -5.5) {\scalebox{.7}{$\left(\begin{array}{cc}
1 & 0\\
r=0 & 1
\end{array}\right)$}};

\node[left] at (7, -2.5) {\scalebox{.7}{$\left(\begin{array}{cc}
1 & r=0\\
0 & 1
\end{array}\right)$}};
\end{tikzpicture}
\caption{\label{fig:5} Obtaining six rays. } 
\end{figure}
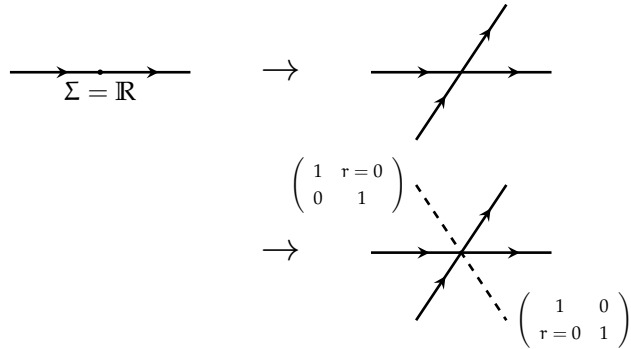

The asymptotic result \eqref{14} is obtained by applying the nonlinear steepest
descent method to the RHP $\lf(\Sg, v_{x,t}\rt)$  in the region $|x| \le c
\;t^{1/3}$.  In this case PII emerges as the RHP $\lf(\Sg, v_{x,t}\rt)$ is
``deformed'' into the RHP $\lf(\Sg, v_x\rt)$ in Figure~\ref{fig:4}.

As we will see, RHPs are useful not only for asymptotics, but also they can
be used to determine symmetries and formulae/identities/equations, and also
for analytical purposes.

\section*{Lecture 2}
\setcounter{equation}{0}\setcounter{section}{2}\addcontentsline{toc}{section}{Lecture 2}

We now consider some of technical issues that arise for RHPs, which are
listed  with bullet points above.

A key role in RH theory is played by the \tit{Cauchy operator}.  We first
consider the case when $\Sg=\RR$.  Here the Cauchy operator $C= C_{\RR}$ is
given by
\beqs
C\;f(z) = \int_{\RR} \frac{f(s)}{s-z}\;\dts, \qquad z \in \CC/\RR,\,
\quad \dts= \frac{ds}{2\pi i}
\eeqs
for suitable functions $f$ on $\RR$ (General refs for the case $\Sg=\RR$, and
also when $\Sg = \{|z|=1\}$, are \cite{Duren} and \cite{Garnett2007}.)  Assume first that
$f\in \mS (\RR)$, the Schwartz space of functions on $\RR$.
Let $z=x+i \, \ep, \; x\in \RR, \; \ep >0$.
Then
\begin{align*}
C\, f(x+i\, \ep) &= \int_\RR \frac{f(s)}{s-x-i\,\ep}\; \dts\\
&= \int_\RR f(s) \frac{i\,\ep}{(s-x)^2 + \ep^2} \;\dts +
\int_\RR f(s) \;\frac{s-x}{(s-x)^2+ \ep^2} \;\dts\\
&=\half \int_\RR f(s) \frac{1}{\pi} \;\frac{\ep}{(s-x)^2+ \ep^2}\;ds
+ \frac{1}{2\pi i} \int_\RR f(s) \; \frac{s-x}{(s-x)^2+ \ep^2}\;ds
\\
&:= I + II.
\end{align*}
Now
\[
I= I_\ep = \half \int_{\RR} \;\frac{f(x+\ep \,u)}{\pi (u^2+1)}\;du.
\]
Then, by dominated convergence,
\beqs
\lim_{\ep \downarrow 0} I_\ep = f(x) \frac{1}{2\pi} \int_{\RR}
\frac{du}{u^2+1} = \half \; f(x).
\eeqs
Write
\[
II = II_\ep = II_{< \ep} + II_{> \ep}
\]
where
\beqs
II_{< \ep} = \frac{1}{2\pi i}\;\int_{|x-s|< \ep } f(s)\;
\frac{s-x}{(s-x)^2 + \ep^2} \;ds
\eeqs
and
\beqs
II_{>\ep} = \frac{1}{2\pi i} \int_{|x-s|> \ep}  f(s) \; \frac{s-x}{(s-x)^2
+ \ep^2} \; ds.
\eeqs
As 
\[
\frac{s-x}{(s-x)^2+ \ep^2}
\]
is an odd function about $s=x$, $II_{<\ep}$
can be written as
\[
II_{< \ep} = \frac{1}{2\pi i} \int_{|s-x| < \ep} \;\frac{s-x}{(s-x)^2 + \ep^2}
\lf(f(s) - f(x) \rt) ds
\]
and so
\begin{align}
\lf| II_{< \ep}\rt| &\le \frac{1}{2\pi} \;\|f'\|_{L^\infty} \int_{|s-x|
< \ep} \: \frac{(s-x)^2}{(s-x)^2+ \ep^2} \;ds \notag \\
&\le\; \|f'\|_{L^\infty} \;\frac{2\ep}{2\pi}  \notag
\end{align}
which goes to $0$ as $\ep \downarrow 0$.  Finally
\begin{align*}
II_{> \ep } &= \frac{1}{2\pi i} \int_{|s-x|> \ep} \lf( \frac{s-x}{(s-x)^2
+ \ep^2} - \frac{1}{s-x} \rt) f(s)\, ds \\
&+ \frac{1}{2\pi i} \int_{|s-x|>\ep} \;\frac{1}{s-x}\;f(s)\, ds \\
&= III_\ep  + IV_\ep.
\end{align*}
We have
\begin{align*}
\lf| III_\ep\rt| &= \frac{1}{2\pi} \lf| \int_{|s-x|>\ep} \;
\frac{\ep^2}{(s-x)^2+\ep^2} \: \frac{f(s)}{s-x} \; ds\, \rt| \\
&= \frac{1}{2\pi} \lf| \int_{|u|>1} \;\frac{1}{u^2+1} \;
\frac{f(x+\ep \,u)}{u} \; du \rt|
\end{align*}
and so as $\ep \downarrow 0$, again by dominated convergence,
\[
| III_\ep | \to \frac{1}{2\pi} \; \lf|f(x)\rt| \lf| \int_{|u|>1}
\;\frac{du}{u(u^2+1)} \rt| =0
\]
as the final integrand is odd.

Thus we see that for $\Sg = \RR$ and $f\in \mS(\RR)$
\beqs
C^+  f(x) \equiv \lim_{\ep \downarrow 0} C f \lf(x+i\,\ep\rt) = \half
\;f(x) + \frac{i}{2} Hf(x)
\eeqs
where
\beqs
Hf(x) = \lim_{\ep \downarrow 0} \frac{1}{\pi} \int_{|s-x|> \ep}
\;\frac{f(s)}{x-s} \;ds.
\eeqs
$H f(x)$ is called the \tit{Hilbert transform} of $f$.
Note that
\begin{align*}
\frac{1}{\pi} \int_{|s-x|> \ep} \; \frac{f(s)}{x-s}\; ds &=
\frac{1}{\pi} \int_{|s-x|>1} \; \frac{f(s)}{x-s}\; ds \\
&+ \frac{1}{\pi} \int_{\ep < |s-x|<1} \;\frac{f(s) - f(x)}{x-s}\;ds
\end{align*}
which converges to
\[
\frac{1}{\pi} \int_{|s-x|>1} \;\frac{f(s)}{x-s}\;ds + \frac{1}{\pi}
\int_{|s-x|<1} \; \frac{f(s)- f(x)}{x-s}\;ds
\]
as $\ep \downarrow 0$, so that $\lim_{\ep \downarrow 0} \; \frac{1}{\pi}
\int_{|s-x|> \ep }\; \frac{f(s)}{x-s}\;ds$ indeed exists pointwise for
$f\in \mS$.

Similarly one finds
\beqs
C^- f(x) \equiv \lim_{\ep \downarrow 0} C f(x-i\ep) = -\half
\; f(x) + \frac{i}{2} \; Hf(x), \qquad x\in \RR
\eeqs
and we obtain the fundamental relations for $f\in \mS$
\beq
C^+ f - C^- f = f \label{48}
\eeq
and
\beqs
C^+ f+ C^- f = i Hf.
\eeqs

\begin{exer}
Show that the limits $C^\pm f(x)= \lim_{\ep \downarrow 0}
Cf(x\pm i\ep)$ are in fact \tit{non-tangential limits} i.e. $C^+ f(x) =
\lim_{z' \to x} \; C f(z')$ where $z'$ lies in a cone of arbitrary
opening angle $\al < \pi$ (see Figure~\ref{fig:6}), and similarly for $C^- f(x)$ (see refs. \cite{Ontheline,Bottcher,ClanceyGohberg}).

\begin{figure}[H]
\centering
\begin{tikzpicture}[scale=1.8]

\filldraw[draw=none, fill = {rgb:black,1;white,6}]  (0,0) -- (1.5,2.4) arc(58:122.5:2.8cm) -- cycle;

\fill (0,0) circle (.06cm);
\node[below] at (0,0){$x$};

\fill (.7, 2.2) circle (.06cm);
\node[right] at (.7,2.2){$z'$};

\draw (-2,0) to (2,0);
\draw (0,0) to (0,3);
\draw(0,0) to (1.5,2.4);
\draw(0,0) to (-1.5,2.4);

\draw[line width = 1, domain=58:122]  plot ({ cos(\x) }, { sin(\x)});
\node[above] at (-.4, 1){$\alpha$};

\end{tikzpicture}
\caption{\label{fig:6} A cone of opening angle $\alpha$.} 
\end{figure}
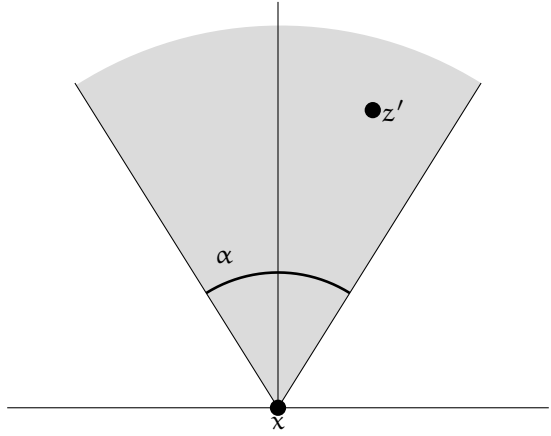

\end{exer}
A critical property of the singular integral operator $H$, and hence the
operators $C^\pm$, is that, as we now show, $H$ is a bounded operator
from $L^p (\RR) \to L^p(\RR)$ for all $1<p < \infty$.
To prove the result for $L^2$, recall that the Fourier transform
\beqs
\hf(z) = \msF\,f(z) \equiv \lim_{R\to\infty} \int^R_{-R} \;e^{-izt}\; f(t)\;
\frac{dt}{\sqrt{2\pi}}, \qquad z\in \RR
\eeqs
and the inverse Fourier transform
\beqs
\chf(x) = \msF^{-1} \, f(x) \equiv \lim_{R\to \infty} \int^R_{-R} e^{izt}\;
f(t) \;\frac{dt}{\sqrt{2\pi}}, \qquad x\in \RR
\eeqs
are unitary maps
\beqs
\|f\|_{L^2} = \|\hf\|_{L^2} = \| \chf\|_{L^2}
\eeqs
from $L^2$ onto $L^2$.  Moreover
\beqs
\msF\; \text{ and }\; \msF^{-1} \quad\text{are indeed inverse to
each other, } 
\eeqs
\[
\lf( \hf \rt)^{\!\!\scriptscriptstyle\vee} = f = \lf(\chf\rt)^{\!\scriptscriptstyle\wedge}\ ,
\qquad f\in L^2(\RR).
\]
For $f\in \mS(\RR)$, fix $\ep > 0$, and set
\beq
\lf(C_\ep \, f\rt) (x) \equiv Cf(x+i \ep ).
\notag 
\eeq
Then
\begin{equation}
	\begin{split}\label{55}
		\lf(\msF\,C_\ep\,f\rt)(z) &= \lim_{R\to\infty} \int^R_{-R} \lf(\int_\RR
		\frac{f(s)}{s-x- i\ep}\;\dts\rt) e^{-ixz}\;\frac{dx}{\sqrt{2\pi}}
		\\
		&= \lim_{R\to\infty} \int_{\RR} \frac{f(s)}{\sqrt{2\pi}}
		\lf( \int^R_{-R} \frac{e^{-ixz}}{s-x-i\ep}\;\dtx\rt) ds,
	\end{split}
\end{equation}
by Fubini's Theorem. 

Now for $s$ fixed and $R$ large, and $z>0$
\begin{align*}
&\int^R_{-R} \frac{e^{-ixz}}{s-x-i\ep}\;\dtx =- \int^R_{-R} \frac{e^{-ixz}}{
x-(s-i\ep)} \; \dtx\\
&\qquad= \int\displaylimits_{\begin{tikzpicture}[scale=.5]
\draw[domain=180:360, directed] plot ({cos(\x)}, {sin(\x)});
\draw[directed] (1,0) to (-1,0);
\node[right] at (1,0) {\scalebox{.7}{$R$}};
\node[left] at (-1,0) {\scalebox{.7}{$-R$}};
\fill (.6, -.3) circle (.06cm);
\node[right] at (1.4, -1) {\scalebox{.7}{$s-i\ep$}};
\draw[->] (1.4, -1) to [out=180, in=270] (.2, -.4) to (.5,  -.3);
\end{tikzpicture}}  \;\frac{e^{-ixz}}{x-(s-i\ep)} \;\dtx - \int \displaylimits_{\begin{tikzpicture}[scale=.5]
\fill (0,0) circle(.06cm);
\draw[domain=180:360, directed] plot ({cos(\x)}, {sin(\x)});
\node[right] at (1,0) {\scalebox{.7}{$R$}};
\node[left] at (-1,0) {\scalebox{.7}{$-R$}};
\end{tikzpicture}}\;
\frac{e^{-ixz}}{x-(s-i\ep)}\;\dtx\\
&\qquad =e^{-i(s-i\ep)z} - \int\displaylimits_{\begin{tikzpicture}[scale=.5]
\fill (0,0) circle(.06cm);
\draw[domain=180:360, directed] plot ({cos(\x)}, {sin(\x)});
\node[right] at (1,0) {\scalebox{.7}{$R$}};
\node[left] at (-1,0) {\scalebox{.7}{$-R$}};
\end{tikzpicture}} \;\frac{e^{-ixz}}{x-(s-i\ep)}\;\dtx.
\end{align*}
\begin{exer} Show that, for $z>0$, we have
\beqs
\lim_{R\to\infty}
\int\displaylimits_{\begin{tikzpicture}[scale=.5]
\fill (0,0) circle(.06cm);

\draw[domain=180:360, directed] plot ({cos(\x)}, {sin(\x)});
\node[right] at (1,0) {\scalebox{.7}{$R$}};
\node[left] at (-1,0) {\scalebox{.7}{$-R$}};
\end{tikzpicture}} \;\frac{e^{-ixz}}{x-(s-i\ep)}\; \dtx =0\,.
\eeqs
\end{exer}
Hence for $s$ fixed and $z>0$
\beqs
\lim_{R\to\infty} \int_{-R}^R \;
\frac{e^{-ixz}}{s-x-i\ep}\; \dtx =e^{-isz}\;e^{-\ep z}.
\eeqs
But we also have

\begin{exer} For $ z>0, $
\[
\int_{-R}^R \;\frac{e^{-ixz}}{s-x-i\ep}\;\dtx
\]
is bounded in $s$ uniformly for $R>0$.
\end{exer}

It follows that we may take the limit $R\to\infty$ in \eqref{55} in the
$s$-integral and so for $z>0$
\beqs
\msF\,C_\ep\,f(z)= \int_{\RR} \,f(s)\, e^{-isz}\;e^{-\ep z} \;
\frac{ds}{\sqrt{2\pi}} = e^{-\ep z} \;\msF f(z).
\eeqs

\begin{exer}
	Show, by a similar argument, that
\beqs
\msF\, C_\ep\, f(z) =0 \quad\text{for} \quad z<0
\eeqs
\end{exer}
\noindent Thus
\beqs
C_\ep \, f = \msF^{-1} \lf(\chi_{>0}(\cdot) \, e^{-\ep\,\cdot}\rt) \msF f
\eeqs
\text{where}
\begin{equation*}
\chi_{>0}(z) \,e^{-\ep z}   = \begin{cases} e^{-\ep z}
    & \text{for~}  z>0; \\
0 & \text{for~} z<0.\end{cases}
\end{equation*}

Now as $\mS(\RR)$ is dense in $L^2$, and as $\msF^{-1} (\chi_{>0}(\cdot) \,
e^{-\ep\cdot})\,\msF$ is clearly bounded in $L^2$ it follows that
\beqs
C_\ep\,f \quad\text{extends to a bounded operator on $L^2$}.
\eeqs
Moreover
\beqs
\hC\, f(x)\equiv \msF^{-1} \, \chi_{>0}(\cdot)\, \msF\,f
\eeqs
is clearly also a bounded operator in $L^2$ and for $f\in L^2$
\begin{align*}
\|C_{\ep}\,f-\hC\,f \|_{L^2} &= \|\msF^{-1}\, \chi_{>0}(\cdot)
\lf(e^{-\ep \,\cdot} -1\rt) \msF\,f\|_{L^2} \\
&= \| \chi_{>0}(\cdot) \lf(e^{-\ep\, \cdot}-1\rt) \msF\, f\|_{L^2}
\end{align*}
which converges to $0$ as $\ep \downarrow 0$, again by dominated convergence.
In other words, for $f\in L^2$,
\beqs \notag
C\,f(x+i\ep) = \int_{\RR} \; \frac{f(s)}{s-x-i\ep} \;\dts \to \hC\, f(x)
\quad\text{in}\quad L^2(dx).
\eeqs
In particular, it follows by general measure theory, that for some sequence
$\ep_n \downarrow 0$
\begin{align}
C\,f(x+i\,\ep_n) \to \hC\,f(x)
\label{eq:ptwse}
\end{align}
pointwise a.e.  In particular \eqref{eq:ptwse} holds for $f\in \mS(\RR)$.  But then by
our previous calculations, $C\,f(x+i\,\ep_n)$ converges pointwise for all
$x$, and we conclude that for $f\in \mS$ and a.e.\ $x$
\beqs
C^+\, f(x) = \half \;f(x) + \frac{i}{2} \; H\, f(x) = \hC\, f(x)
\qquad\qquad \text{a.e.}\; x.
\eeqs
Thus $C^+\,f$ and, hence $H\,f$, extend to  bounded operators on $L^2(\RR)$
and
\begin{align*}
\half\;\hf+ \frac{i}{2} \; \msF\,H\, f&= \msF\,\hC\,f= \chi_{>0}\,\hf \\[-5mm]
\intertext{and so}
\msF\,H\,f(z) &= \frac{2}{i} \lf(\chi_{>0}(z)-\half\rt)
\hf(z)\\ 
&= -i\, \sgn \, (z) \;\hf(z)
\end{align*}
where $\sgn (z) =+ 1$ if $z>0$ and $\sgn(z) =-1$ if $z<0$.

We have shown the following: For $f\in L^2$,
\beqs
C^+\,f= \frac{f}{2} + \frac{i}{2}\;H\,f = \frac{f}{2} + \half
\lf(\hf \; \sgn(\cdot)\rt)^{\vee}
\eeqs
and similarly
\beqs
C^-\,f = - \frac{f}{2} + \frac{i}{2}\;H\,f =- \frac{f}{2} + \half
\lf(\hf \, \sgn (\cdot)\rt)^{\vee}.
\eeqs

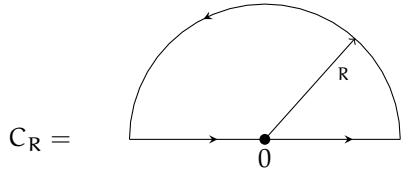
\begin{figure}[H]
\centering
\begin{tikzpicture}[scale=1.2]
\node at (0,0) {$C_R =$};
\fill (2.5,0) circle (.06cm);
\node[below] at (2.5,0) {$0$};
\draw[directed]  (1,0) to (2.5,0);
\draw[directed] (2.5, 0) to (4,0);
\draw[->] (2.5,0) to (3.5, 1.12);
\draw[domain=0:180,directed] plot ({2.5 +1.5*cos(\x)}, {1.5*sin(\x)});
\node[below right] at (3.2, .9) {\scalebox{.7}{$R$}};
\end{tikzpicture}
\caption{A semi-circle in the upper-half plane.}\label{7} 
\end{figure}

The following argument of Riesz shows that in fact $C^\pm$, and hence $H$,
are bounded in $L^p(\RR)$, for all $1< p < \infty$.  Consider first the case $p =4$. Suppose
$f\in \msC^\infty_0 (\RR)$,  the infinitely differentiable functions with
compact support.
Then as $z\to \infty$,
\[
C\,f(z) = \int_{\RR} \frac{f(s)}{s-z}\;\dts = O\lf(\frac{1}{z}\rt)
\]
and $C\,f(z)$ is continuous down to the axis.  By Cauchy's theorem
\[
\int_{C_R}  \lf(C\,f(z)\rt)^4 \, dz =0
\]
where $C_R$ is given in Figure \ref{7},
and as 
\[\int\limits_{\kern-3pt\begin{tikzpicture}[scale = .3]
\draw [domain=0:180,directed] plot ({-0.25+cos(\x)}, {-0.3+sin(\x)});
\node[right] at (0.75,-0.3){\scalebox{.5}{$R$}};
\node[left] at (-1.25,-0.3){\scalebox{.5}{$-R$}};
\end{tikzpicture}}\;\kern-11pt \lf(C\,f(z)\rt)^4 \,dz\to 0
\text{~~as~~} R\to \infty,\] 
we conclude that
\[
\int_{\RR} \lf(C^+\,f(x)\rt)^4 \, dx =0.
\]
But then as $C^+\,f= \frac{f}{2} + \frac{i}{2} \;H\, f$ we obtain
\beq
0= \int_{\RR} \lf(f^4 + 4 f^3 (Hf)\,i + 6f^2 (Hf)^2\,i^2 + 4f(Hf)^3\,i^3
+(Hf)^4\,i^4\rt) dx.
\label{65}
\eeq
Now suppose that $f$ is real.  Then $Hf$ is real and the real part
of \eqref{65} yields
\[
0= \int_{\RR} \lf(f^4-6 f^2 (Hf)^2 + (Hf)^4\rt)dx,
\]
hence
\begin{align}\notag
	\int_{\RR}(Hf)^4\,dx  &= 6 \int_{\RR} f^2(Hf)^2dx - \int_{\RR} f^4 dx\notag\\
&\le 6 \lf(\int \frac{c}{2} \; f^4 dx + \int_{\RR} \frac{1}{2c}\; (Hf)^4 dx\rt)
-\int_{\RR} f^4 dx\notag
\end{align}
for any $c>0$.  Take $c=6$.  Then
\begin{align*}
\half \int_{\RR} (Hf)^4 dx&\le (18-1) \int_{\RR} f^4dx\\
\intertext{or}
\int_{\RR} (Hf)^4 dx&\le 34 \int_{\RR} f^4dx.
\end{align*}
The case when $f$ is complex valued is handled by taking real and imaginary
parts.  Thus, by density, $H$ maps $L^4$ boundedly to $L^4$.

\begin{exer}
Show that $H$ maps $L^p\to L^p$ for all $1< p < \infty$.
Hints:
\begin{enumerate}
\item Show that the above argument works for
all even integers $p$.
\item Show that the result follows for all $p\ge 2$ by
interpolation.
\item Show that the result for $1<p<2$ now follows by duality.
\end{enumerate}
\end{exer}

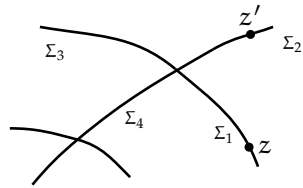
\begin{figure}[H]
\centering
\begin{tikzpicture}[scale =1.]
\fill (2.28,-.1) circle (0.06cm);
\node[right] at (2.28, -.1) {$z$};

\fill (2.3, 1.4) circle (0.06cm);
\node[above] at (2.3,1.4) {$z'$};

\draw[line width = 1] (0,0) to [out=220, in=50](-.6,-.6);
\draw[line width = 1] (0,0) to [out=165, in=0](-.9,.15);
\draw[line width = 1] (0,0) to [out=-15, in=135](.7,-.5);
\draw[line width = 1] (1.8,1.2) to [out=210, in = 40] (0,0);
\draw[line width = 1] (1.3,.94) to [out=140, in = -10] (-.5,1.5);
\draw[line width = 1] (1.3,.94) to [out=-40, in = 110] (2.4,-.36);
\draw[line width = 1] (1.8,1.2) to [out=30, in = 200] (2.6,1.5);

\node[below left] at (2.2, .3) {\scalebox{.7}{$\Sigma_1$}};
\node[below right] at (2.6, 1.5) {\scalebox{.7}{$\Sigma_2$}};
\node[below] at (-.3, 1.4) {\scalebox{.7}{$\Sigma_3$}};
\node[below right] at (.5, .5) {\scalebox{.7}{$\Sigma_4$}};
\end{tikzpicture}
\caption{\label{fig:8} Contours that self-intersect. } 
\end{figure}

\begin{exer} Show that $H$ is \tit{not} bounded from
$L^1 \to L^1$.
(However $H$ maps $L^1 \to \text{ weak } L^1$.)
As indicated in Lecture 1, RHPs take place on contours which self-intersect (see Figure~\ref{fig:8}).
\end{exer}

We will need to know, for example, that if $f$ is supported on $\Sg_1$, say,
and we consider
\[
Cf(z') = \int_{\Sg_1}\;\frac{f(z)}{z-z'}\;\dtz
\]
for $z'\in \Sg_2$, say, then $Cf(z') \in L^2(\Sg_2)$
if $f\in L^2(\Sg_1)$. Here is a prototype result which one can prove using the
Mellin transform, which we recall is the Fourier transform for the
multiplicative group $\{x>0\}$.  We have \cite[p.~88]{Ontheline}  the following:

For $f\in L^2(0, \infty)$ and $r>0$, set
\beqs
C_\theta\, f(r) = \int^\infty_0 \,\frac{f(s)}{s-\hz\,r} \;\dts, \qquad
\hz=e^{i\theta} 
\eeqs
where $0<\theta <2\pi$.  Then
\beq
\|C_\theta\,f\|_{L^2(dr)} \le c_\theta \|f\|_{L^2(ds)}
\label{68}
\eeq
where
\beqs
c_\theta = \ga^\ga\; (1-\ga)^{1-\ga}, \qquad \ga= \frac{\theta}{2\pi}.
\eeqs

One can also show that for any $1<p < \infty$
\beqs
\|C_\theta\, f\|_{L^p(dr)} \le C_{\theta, p} \,\|f\|_{L^p (ds)}
\eeqs
for some $c_{\theta, p} < \infty$.

Results such as \eqref{68} are useful in many ways.  For example, we have the
following result.

\begin{theorem}\label{thm2}
Suppose $f\in H^1(\RR)= \{ f\in L^2(\RR): f'\in  L^2(\RR)\}$.  Then $C\,f(z)$
is uniformly H\"{o}lder-$\half$ in $\ovl{\CC}_+$ and in $\ovl{\CC}_-$.
In particular, $Cf$ is continuous down to the axis in $\CC_+$ and in $\CC_-$.
\end{theorem}
\begin{proof}
For $z\in \CC \setminus \RR$
\begin{align*}
\frac{d}{dz} \;Cf(z) &= \int_{\RR} \frac{f(s)}{(s-z)^2} \;\dts =- \int_{\RR}
\lf( \frac{d}{ds} \lf(\frac{1}{s-z}\rt)\rt) f(s) \;\dts\\
&= \int_{\RR} \frac{f'(s)}{s-z} \; \dts.
\end{align*}
Now suppose $z',\, z'' \in \CC_+$, and the straight line $L$ through $z',z''$ intersects the line $\RR$ at $x$
at an angle $\theta$ as in Figure~\ref{fig:9}.
\begin{figure}[H]
\centering
\begin{tikzpicture}[scale=1.]
\fill (0,0) circle (.06cm);
\node[below] at (0,0) {$x$};
\fill (.6,.6) circle (.06cm);
\node[right] at (.6,.6) {$z'$};
\fill (1.2,1.2) circle (.06cm);
\node[right] at (1.2,1.2) {$z''$};

\node[above right] at (.2, -0.03) {\scalebox{.8}{$\theta$}};

\node[right] at (2.2, 2) { $L$};

\draw (-2,0) to (2,0);
\draw (0,0) to (2,2);

\end{tikzpicture}
\caption{\label{fig:9} A line intersecting $\mathbb R$.} 
\end{figure}
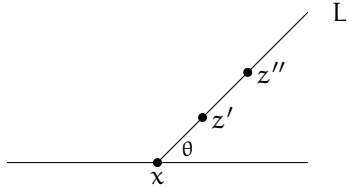
Then as $\int_\RR \frac{f'(s)}{s-z}\;\dts
= \int^\infty_x \;\frac{f'(s)}{s-z}\;\dts + \int^x_{-\infty} \;
\frac{f'(s)}{s -z}\;\dts$, and as $f'\in L^2(-\infty, x)
\oplus L^2(x,\infty)$  it follows from \eqref{68} that
\[
\int_0^\infty \lf|\frac{d}{dz} \;Cf\lf(re^{i\theta}\rt) \rt|^2\,dr=
\int^\infty_0 \lf|C\,f' (re^{i\theta})\rt|^2\, dr \le c \|f'\|_{L^2}.
\]
But
\begin{align*}
&\lf| Cf(z'') - Cf(z') \rt| = \lf|\int_{z'\to z'' \text{ in } L}
\,\frac{d}{dz}\;Cf (z) \, dz \rt|\\
&\quad \le \lf|z''-z' \rt|^{\half} \| \frac{d}{dz} \;Cf\|_{L^2(0,\infty)}
\le c|z''-z'|^{\half}\;\|f'\|_{L^2(\RR)}.\qedhere
\end{align*}
\end{proof}

We now consider general contours $\Sg \subset \ovl{\CC}= \CC \cup
\{\infty\}$, which are \tit{composed curves}:  By definition a composed curve $\Sg$ is a finite union of arcs
$\{\Sg_i\}_{i=1}^n$ which can intersect only at their end points.  Each
arc $\Sg_i$ is  homeomorphic to an interval $[a_i, b_i] \subset \RR$:
\[
\lf\{
\begin{array}{lll}
\!\!\!\!\!\!&\varphi_i : [a_i, b_i]\!\!
&\to \Sg_i \subset \ovl{\CC},\\
\!\!\!\!\!\!&[a_i, b_i] \ni t  \!\!
&\to \varphi_i(t) \in \Sg_i, \, \varphi_i(a_i)
\neq \varphi_i(b_i).
\end{array} \rt.
\]
Here $\ovl{\CC}$ has the natural topology generated by the sets $\{|z|< R_1\},
\;\;\{|z| > R_2\}$ where $R_1, \, R_2>0$.  A loop, in particular the unit circle
$T=\{|z|=1\}$, is a composed curve on the understanding that it is a union
of (at least) two arcs.

Although it is possible, and sometimes useful, to consider other function
spaces (e.g.\ H\"{o}lder continuous functions), we will only consider RHPs
in the sense of $L^p(\Sg)$ for $1< p < \infty$.  

So the first question is
``What is $L^p(\Sg)$?''.  The natural measure theory for each arc
$\Sg_i$ is generated by arc length measure $\mu$ as follows.
If $z_0= \varphi(t_0)$ and $z_n = \varphi(t_n)$ are the end-points of
some arc $\Sg \subset \CC$, and $z_0,\, z_1, \dots, \,z_n$ is any
partition of $[z_0, z_n]= \{\varphi(t): t_0 \le t \le t_n\}$
(we assume $z_{i+1}$ succeeds $z_i$ in the ordering induced on $\Sg$
by $\varphi$, symbolically $z_i< z_{i+1},\, \text{etc.})$
then
\beqs
L=L_{[z_0, z_n]} \equiv \sup_{\text{all partitions $\{z_i\}$}} \quad
\sum_{i=0}^{n-1} \lf|z_{i+1}-z_i\rt|.
\eeqs
If $L<\infty$ we say that the arc $\Sg=\lf[z_0, z_n\rt]$ is \tit{rectifiable}
and $L_{[z_0, z_n]}$ is its \tit{arc length}.  We will only consider composed
curves $\Sg$ that are \tit{locally rectifiable} i.e.\ for any $R>0$,
$\Sg \cap \{|z| < R\}$ is rectifiable (note that the latter set is an at most
countable union of simple arcs and rectifiability of the set means that the
sum of the arc lengths of these arcs is finite.  In particular, the unit circle
$\TT$ as a union of $2$ rectifiable subarcs, is rectifiable,
and $\RR$ is locally rectifiable.)
For any interval $[\al, \beta)$ on $\Sg_i \subset \CC$ (the case where
$\Sg_i$ passes through $\infty$, must be treated separately --- exercise!)
define
\[
\mu_i \lf([ \al, \beta)\rt) = \text{ arc length } \al \to \beta.
\]
Now the sets $\{ [\al, \beta): \al < \beta \text{ on } \Sg_i\}$ form a
semi-algebra (see \cite{Royden}) and hence $\mu_i$ can be extended
to a complete measure on a $\sg$-algebra  $\msA$ containing the Borel sets
on $\Sg_i$. The restriction of the measure to the Borel sets is unique.  For $1 \le p < \infty$, we
can define $L^p \lf(\Sg_i, d\mu_i\rt)$ to be the set of $f$ measurable with respect to $\msA$
on $\Sg$ for which,
\begin{equation*}
 \int_{\Sg_i} \lf|f(z)\rt|^p \, d\mu_i(z) < \infty ,
\end{equation*}
and then all the ``usual'' properties go through.  One usually writes $d\mu=|dz|$.
For $\Sg=\bigcup^n_{i=1}\, \Sg_i,\;L^p(\Sg, d\mu)$ is simply the
\tit{direct sum} of $L^p\lf(\Sg_i, d\mu_i\rt)^n_{i=1}$.

\begin{exer} $|dz|$ is also  equal to Hausdorff-1 measure on
$\Sg_1$.
\end{exer}

Note that if $\Sg_1=\RR$ and $\Sg_2 = \lf\{\lf(x, x^3 \,\sin \, \frac{1}{x}\rt)
: x\in \RR\rt\}$ then $\Sg= \Sg_1 \cup \Sg_2$ is not a composed curve,
although $\Sg_1$ and $\Sg_2$ are both locally rectifiable.

For $\Sg$ as above we  define the Cauchy operator for $h\in L^p
\lf(\Sg, |dz|\rt)$, $1 \le p < \infty$, by
\begin{equation}\label{72}
Cf(z) = C_\Sg\,f(z) = \int_\Sg \:\frac{f(\zeta)}{\zeta-z} \;\ddbar \zeta,
\qquad z\in \CC\backslash \Sg.
\end{equation}
Given the homeomorphisms $\varphi_i: \lf[a_i, b_i\rt] \to \Sg_i$, the contour
$\Sg$ carries a natural orientation, and the integral here is a line integral
following the orientation; if we parametrize the arcs $\Sg_i$ in $\Sg$ by
arc length $s$,
\[
0 \le s \le s_i,\qquad \zeta=\zeta(s), \qquad\text{then}\qquad
\lf|\frac{d\zeta(s)}{ds} (s) \rt|=1 \quad \text{(why?)}
\]
and \eqref{72} is a sum over its subarcs $\Sg_i$ of integrals of the form
\[
\int_0^{s_i} \frac{f\lf(\zeta(s)\rt)}{\zeta (s) -z} \;\frac{d\zeta(s)}{ds}
\;\dts, \qquad z\in \CC\backslash \Sg
\]
for each $i$, the integrand (clearly) lies in $L^p \lf(ds: [0, s_i)\rt)$.

Now the fact of the matter is that many of the properties that were true
for $C_\Sg$ when $\Sg=\RR$, go through for $C_\Sg$ in the general situation.
(See, in particular, \cite{Goluzin1969}.)
In particular for $f\in L^p (\Sg, d\mu)$, the non-tangential limits
\beq
C^\pm_\Sg\, f(z) = \lim_{z'\to z^\pm} \; C_\Sg \, f(z')
\label{73}
\eeq
exist pointwise a.e.\ on $\Sg$.  Figure~\ref{fig:12} demonstrates non-tangential limits.
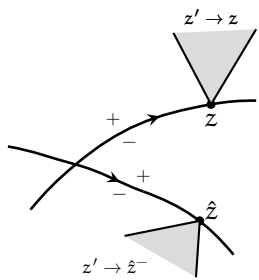
\begin{figure}[H]

\centering
\begin{tikzpicture}[scale =1]
\draw[line width = 1] (0,0) to [out=220, in=50](-.6,-.6);
\draw[line width = 1] (0,0) to [out=165, in=-10](-.9,.24);
\draw[line width = 1, directed] (0,0) to [out=-15, in=160](.9,-.36);
\node[above] at (.9,-.36) {\scalebox{.7}{$+$}};
\node[left] at (.8,-.4) {\scalebox{.7}{$-$}};

\draw[line width = 1](.9,-.36) to [out=-20, in= 130](2.1,-1.2);
\fill (1.65, -.75) circle (.06cm);
\node[above right] at (1.6, -.85) {$\hat{z}$};

\draw[line width = 1, directed] (0,0) to [out=40, in = 190] (1.8, .8);
\node[above] at (.5,.4) {\scalebox{.7}{$+$}};
\node[below right] at (.5,.5) {\scalebox{.7}{$-$}};

\draw[line width = 1] (1.8, .8) to [out=10, in = 180] (2.4, .85);
\fill (1.8, .8) circle (.06cm);
\node[below] at (1.8, .8) {$z$};

\filldraw[draw=none, fill = {rgb:black,1;white,6}]  (1.8, .8) -- (2.4, 1.8) -- (1.3, 1.75)-- cycle;
\draw[thick] (1.8, .8) to (2.4, 1.8);
\draw[thick] (1.8, .8) to (1.3, 1.75);

\node[above] at (1.8, 1.75) {\scalebox{.7}{$z' \to  z$}};

\filldraw[draw=none, fill = {rgb:black,1;white,6}]  (1.65, -.75) -- (.65, -.95) -- (1.6, -1.5)-- cycle;
\draw[thick] (1.65, -.75) to (.65, -.95);
\draw[thick] (1.65, -.75) to (1.6, -1.5);

\node[below left] at (1.1, .-1.12) {\scalebox{.7}{$z' \to \hat{z}^{-}$}};

\node[above] at (1.8, 1.75) {\scalebox{.7}{$z' \to  z$}};

\end{tikzpicture} 
\caption{\label{fig:12}  Non-tangential limits.} 
\end{figure}
Note that as $\Sg_i$ is locally rectifiable, the tangent vector to the arc
$\frac{d\zeta}{ds}$ exists at a.e.\ point $\zeta=\zeta(s)$: the normal to
$\frac{d\zeta}{ds}$ bisects the cone.
\begin{figure}[H]
\centering
\begin{tikzpicture}[scale=0.6]
\draw[line width=1]  (-1,2) to (1,-2);

\draw[line width=1] ({-3.5/sqrt(5)}, {2.5/sqrt(5)}) to (0,0);
\draw[line width=1]  (0,2) to (0,0);

\node[left] at (-1,2){\scalebox{1}{The normal}};
\node[] at (-0.1,-0.5) {$\xi$};

\coordinate (B) at (0,0);
\coordinate (C) at (-2,-1.3);
\coordinate (A) at (2,0.7);
\arcThroughThreePoints{A}{B}{C};
\end{tikzpicture} 
\caption{\label{fig:13}  A contour and its normal.  } 
\end{figure}
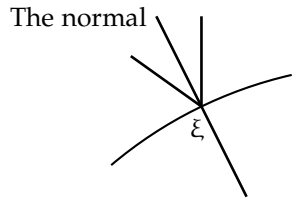
Moreover,
\beqs
C^\pm_\Sg\, f(z) = \pm \half \;f(z) + \frac{i}{2}\:Hf(z)
\eeqs
where the Hilbert transform is now given by
\beq
Hf(z) = \frac{1}{\pi} \,\lim_{\ep \downarrow 0}\; \int_{\begin{subarray}{l}
|s-z|>\ep \\
s\in \Sg
\end{subarray} }
\; \frac{f(s)}{z-s} \; ds, \qquad z\in \Sg
\label{75}
\eeq
and the points $z\in \Sg$ for which the non-tangential limits \eqref{73} exists
are precisely the points for which the limit in \eqref{75} exists.

Again, for $f\in L^p(\Sg, d\mu)$ with $1 \le p < \infty$,
\begin{align}
C^+\,f (z) - C^-\,f(z) &= f(z) \notag \\
\intertext{and}
C^+\,f(z) + C^-\,f(z) &= i\,Hh(z). \notag
\end{align}
The following issue is crucial for the analysis of RHPs:
\begin{ques} 
For which locally rectifiable contours $\Sg$ are the operators $C^\pm$ and $H$
bounded in $L^p, \;1<p<\infty$?
\end{ques}

Quite remarkably, it turns out that there are \tit{necessary} and
\tit{sufficient} conditions on a simple rectifiable curve for $C^\pm,\,H$
to be bounded in $L^p(\Sg), \;1<p<\infty$.  The result is due to many
authors, starting with Calder\'{o}n \cite{C}, and then Coifman, Meyer and McIntosh \cite{CMM},
with Guy David \cite{D} (see \cite{Bottcher} for details and historical references) making the final decisive contribution.

Let $\Sg$ be a simple, rectifiable curve in $\CC$.  For any $z\in \Sg$, and
any $r>0$, let
\[
\ell_r (z) = \text{ arc length of }\lf(\Sg \cap D_r(z)\rt)
\]
where $D_r(z)$ is the ball of the radius $r$ centered at $z$, see Figure~\ref{fig:14}.
\begin{figure}[H]
\centering
\begin{tikzpicture}[scale=0.7]
\fill (0,.05) circle (.06cm);
\node[below right] at (0,0.2) {$z$};
\draw[line width=1] (-2,-2) to [in=90] (.1, .2);
\draw[directed, line width=1] (.1,.2) to [out=270, in=90] (-.4,-.3) to [out=270, in=180] (.15, -.3) to [out=0, in=90] (.2, -.7) to [out = 0, in = 270] (.8, -.6)
   to [out = 90, in = 225] (.45, .2) to [out = 45] (3, 1);

\draw[line width = 1, domain = 0:360] plot ({.7* cos(\x)}, {.7 * sin(\x)});
\node [above left] at (0,.7) {$D_r(z)$};
\node[left] at (-1.5, -.9) {$\Sigma$};
\end{tikzpicture} 
\caption{\label{fig:14} A ball $D_r(z)$ of radius $r$ centered at $z$.} 
\end{figure}
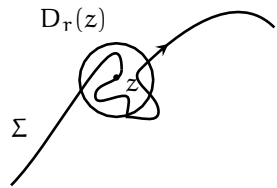
Set
\beqs
\la= \la_\Sg = \sup_{z\in \Sg, \;r>0}\;\frac{\ell_r(z)}{r}.
\eeqs
\begin{theorem}
Suppose $\la_\Sg < \infty$.  Then for any $1<p<\infty$, the limit in \eqref{75}
exists for a.e.\ \, $z\in \Sg$ and defines a bounded operator for any
$1<p<\infty$
\beq
\|H\,f\|_{L^p} \le c_p \|f\|_{L^p}, \quad f\in L^p, \quad c_p< \infty.
\label{79}
\eeq
Conversely if the limit in \eqref{75} exists a.e.\ and defines a bounded operator
$H$ in $L^p(\Sg)$ for some $1<p<\infty$, then $H$ gives rise to a bounded
operator for all $p$, $1<p<\infty$, and $\la_\Sg < \infty$.
\end{theorem}

An excellent reference for the above Theorem, and more, is \cite{Bottcher}.



\begin{remarks}Additional remarks:
\begin{enumerate}
	\item Locally rectifiable curves $\Sg$ for which
$\la=\la_\Sg < \infty$ are called Carleson curves,
	\item the constant  $c_p$ in \eqref{79} has the form $c_p= \phi_p (\la_\Sg)$
for some continuous, increasing function, $\phi_p(t)\ge 0$,
\tit{independent} of $\Sg$, such that $\phi_p(0)=0$.
\end{enumerate}
The fact that $\phi_p$ is independent of $\Sg$, is very important for the
nonlinear steepest descent method, where one deforms curves in a similar
way to the classical steepest descent method for integrals.
\end{remarks}

Carleson curves are sometimes called AD-regular curves: the A and D denote Ahlfors and David. To get some sense of the subtlety of the above result, consider the following
curve $\Sg$ with a cusp at the origin (see Figure~\ref{fig:16}):
\[
\Sg= \{0 \le x \le 1,\;\; y=0\} \cup \{ (x, x^2): \;0 \le x \le 1\}.
\]
\begin{figure}[H]
\centering
\begin{tikzpicture}
\draw[line width=1, directed] (3,0) to (0,0);
\draw[line width =1, domain = 0:1, directed] plot (3*\x, {2.5*(\x^2)});
\node[right] at (3,0) {$(1,0)$};
\node[left] at (0,0) {$0$};
\node[right] at (3,2.5) {$(1, 1)$};
\node at (1, 2) {\scalebox{1.2}{$\Sigma$}};
\end{tikzpicture} 
\caption{\label{fig:16} A cusp at the origin.} 
\end{figure}
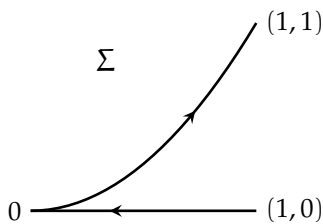
Clearly $\la_\Sg < \infty$ so that the Hilbert transform $H_\Sg$
is bounded in $L^p$, $1< p< \infty$.

\begin{exer} For $\Sigma$ in Figure~\ref{fig:16}, prove directly that $H_\Sg$ is bounded in $L^2$.  The presence of the cusp makes the proof surprisingly difficult.
\end{exer}

\section*{Lecture 3}
\setcounter{equation}{0}\setcounter{section}{3}\addcontentsline{toc}{section}{Lecture 3}

We now make the notion of a RHP precise (see \cite{ClanceyGohberg,deiftzhounls,Litvinchuk1987}).
Let $\Sg$ be a composite, oriented Carleson contour in $\CC$ and let
$v:\Sg \to \textup{GL} (n, \CC)$ be a jump matrix on $\Sg$, with $v, \, v^{-1}
\in L^\infty(\Sg)$.
Let $Ch(z)= C_\Sg h(z), \;C^\pm_\Sg\, h,\; H_\Sg\, h$ be the associated
Cauchy and Hilbert operators.

We say that a pair of $L^p(\Sg)$ function $f_\pm \in \p C(L^p)$ if
there exists a (unique) function $h\in L^p(\Sg)$ such that
\beqs
f_\pm (z) =(C^\pm\,h)(z), \qquad z\in \Sg.
\eeqs
In turn we call $f(z) \equiv Ch(z)$, $\;z\in \CC\backslash \Sg$, the
\tit{extension} of $f_\pm = C^\pm \,h \in \p C(L^p)\;$ \tit{ off }~$\Sg$.

\begin{define}\label{def1}
Fix $1< p < \infty$.  Given $\Sg,\; v$ and a measurable function $f$ on $\Sigma$, we say
that $m_\pm \in f + \p C(L^p)$ solves an \tit{inhomogeneous RHP
of the first kind} ($IRHP1_p$) if
\beqs
m_+(z) =m_-(z) \, v(z), \qquad z\in \Sg.
\eeqs
\end{define}

\begin{define}\label{def2}
Fix $1<p<\infty$.  Given $\Sg,\;v$ and a function $F\in L^p(\Sg)$, we say
that $M_\pm \in \p \,C(L^p)$ solves an \tit{inhomogeneous RHP of the
second kind} $(IRHP2_p)$ if
\beqs
M_+(z) = M_-(z) \;v(z) + F(z), \qquad z\in \Sg.
\eeqs
\end{define}

Recall that $m$ solves the normalized RHP $(\Sg, v)$ if, at least formally,
\begin{align}\label{84}
\begin{aligned}
&\bullet \quad m(z) \quad\text{is a $n\times n$ analytic function in } \CC
\backslash \Sg,\\
&\bullet \quad m_+(z) = m_-(z) \; v(z), \qquad z\in  \Sg,\\
&\bullet \quad m(z) \to  I\quad \text{ as }\quad z\to \infty.
\end{aligned}
\end{align}

More precisely, we make the following definition.
\begin{define}\label{def3}
Fix $1<p < \infty$.  We say that $m_\pm$ solves the normalized RHP
$(\Sg, v)_p$ if $m_\pm$ solves the $IRHP1_p$ with $f\equiv I$.
\end{define}

In the above definition, if $m_\pm -I \in C^\pm\,h$, then clearly the
extension
\[
m(z) = I + Ch(z), \qquad z\in \CC\backslash \Sg,
\]
off $\Sg$ solves the normalized RHP in the formal sense of \eqref{84}.

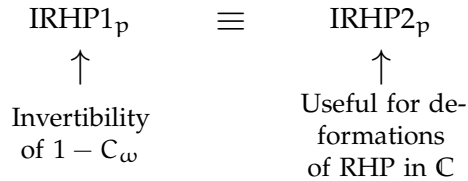
\begin{figure}[H]
\centering
\begin{tikzpicture}[scale = 1]
\node at (0,0) {\scalebox{1.5}{$\equiv$}};
\node at (2,0) {\scalebox{1.1}{IRHP$2_p$}};
\node at (-2,0) {\scalebox{1.1}{IRHP$1_p$}};
\node at (-2, -0.6) {\scalebox{1.5}{$\uparrow$}};
\node at (2, -0.6) {\scalebox{1.5}{$\uparrow$}};
\node[text width = 3cm, align=center] at (-2,-1.5){Invertibility of $1 - C_\om$};
\node[text width = 3cm, align=center] at (2,-1.5){Useful for deformations of RHP in $\mathbb{C}$};
\end{tikzpicture}
\caption{\label{fig:18}  The uses of $IRHP1_p$ and $IRHP2_p$. } 
\end{figure}

Let
\beqs
\begin{aligned}
v ={}& \lf(v^-\rt)^{-1} \, v^+ = \lf(I-\om^-\rt)^{-1} \lf(I+ \om^+\rt) \\
&\quad \om^+ \equiv v^+ -I, \qquad \om^- \equiv I-v^-,
\end{aligned}
\eeqs
be a \tit{pointwise a.e.} factorization of $v$, i.e., $v(x) =  \lf(v^-\rt(x))^{-1} \, v^+(x)$ for a.e. $x$, with $v^\pm, \;(v^\pm)^{-1}
\in L^\infty$, and let $\om = (\om^-, \; \om^+)$.  Let $C_\om$ denote
the basic associated operator
\beqs
C_\om \, h \equiv C^+(h\,\om^-) + C^-(h\, \om^+)
\eeqs
acting on $L^p(\Sg)-n \times n$ matrix valued functions $h$.
As $\om^\pm \in L^\infty$, $\; C_\om \in \mL(L^p)$, the
bounded operators on $L^p$, for all $1<p<\infty$.
The utility of $IRHP1_p$ and $IRHP2_p$ will soon become clear, see Figure~\ref{fig:18}.

\begin{theorem}\label{thm4}
If $f$ and $v$ are such that $f(v-I)\in L^p(\Sg)$ for some $<p< \infty$, then
\beqs
m_\pm = M_\pm + f
\eeqs
solves $IRHP1_p$ if $M_\pm$ solves $IRPH2_p$ with $F=f(v-I)$.  Conversely
if $F\in L^p(\Sg)$, then
\beqs
M_+ = m_+ + F, \qquad M_- = m_-
\eeqs
solves $IRHP2_p$ if $m_\pm$ solves $IRHP1_p$ with $f= C^-\, F$.
\end{theorem}

The first part of this result is straightforward:
Suppose $M_\pm \in \p C(L^p)$ solves $
M_+ = M_- \, v+ F $ on $\Sg$
with $F= f(v-I) \in L^p$.  Then
\[
M_+ = M_-v + f(v-I) 
= (M_- + f) v-f
\]
or
$m_+ = m_- v$
with
$m_\pm = f + M_\pm \in f + \p C(L^p)$.
The converse is more subtle and is left as an exercise:

\begin{exer}\label{exer10}
Show $IRHP1_p \Rightarrow IRHP2_p$.
\end{exer}

We now show that the RHPs $IRHP1_p$ and $IRHP2_p$, and, in particular, the
normalized RHP $(\Sg, v)_p$ are intimately connected with the singular
integral operator $1-C_\om$.

Let $f\in L^p(\Sg)$ and let $m_\pm = f+C^\pm \, h$ for some $h\in L^p(\Sg)$.
Also suppose $m_+ = m_-\, v = m_-(v^{-})^{-1}\, v^+$.  Set
\[
\mu= m_- \lf(v^{-}\rt)^{-1} = m_+ \lf(v^+\rt)^{-1} \in L^p(\Sigma)
\]
and define
\[
H(z) = \lf(C \lf(\mu \lf(\om^+ + \om^{-}\rt)\rt)\rt)(z), \qquad z\in \CC
\backslash \Sg.
\]
Then we have on $\Sg$, using $C^+ - C^- = 1$
\begin{align*}
H_+ &= C^+ \lf(\mu \lf(\om^+ + \om^-\rt) \rt) = C^+ \mu\, \om^+ +
C^+\mu\,\om^-\\
&= C^+\mu\,\om^- + C^- \mu\, \om^+ + \mu \,\om^+
\\
&= C_\om\, \mu + \mu\,\om^+ =
\lf( C_\om-1\rt)\mu+ \mu \lf(I + \om^+\rt) \\
&= \lf(C_\om - 1\rt) \mu + \mu\, v^+ = \lf(C_\om-1\rt)\mu +  m_+.\\
\text{i.e. }\; H_+ &= \lf(C_\om -I\rt) \mu + m_+.
\end{align*}
Similarly
\[
H_- = \lf(C_\om - 1\rt) \mu + m_-.
\]
Thus
\beq
m_\pm - f - H_\pm = \lf( 1-C_\om \rt) \mu-f.
\label{89}
\eeq
But $m_\pm -f - H_\pm \in \p C\, (L^p)$; i.e.
$ m_\pm - f - H_\pm = C^\pm \, h$ for some $h\in L^p$.
However, from \eqref{89}
\[C^+\,h = C^-\,h \qquad
\Rightarrow \qquad h= C^+\,h - C^-\, h =0.
\]
We conclude that $(1-C_\om)\,\mu=f, \mu\in L^p$.

\noi
Conversely, if $\mu \in L^p(\Sg)$ solves $(1-C_\om)\,\mu=f, \;$ then
the above calculations show that $H\equiv C \lf(\mu \lf(\om^+ + \om ^-\rt)\rt)$
satisfies
\[
H_\pm =-f + \mu \,v^\pm.
\]
Thus setting $m_\pm = \mu\, v^\pm$, we see that $m_+ = m_-\, v$ and
$m_\pm -f \in \p \,C(L^p)$.  In particular $\mu \in L^p$ solves
$(1- C_\om)\,\mu=0$ iff $ m_\pm = \mu\,v^\pm$ solves the
homogeneous~RHP.\\[-\baselineskip]
\begin{equation}\label{90}
m_+ = m_-\, v, \qquad m_\pm \in \p C (L^p).
\end{equation}
We summarize the above calculations as follows:

\begin{prop}\label{prop1}
Let $1<p< \infty$.  Then
\begin{align*}
(1- C_\om) &\quad\text{ is a bijection in } \;L^p(\Sg)\\
&\Longleftrightarrow\\
IRHP1_p &\quad\text{ has a unique solution for all}
\quad f\in L^p(\Sg)\\
&\Longleftrightarrow\\
IRHP2_p &\quad\text{ has a unique solution for all}
\quad F\in L^p(\Sg).
\end{align*}
Moreover, if one, and hence all three of the above conditions, is satisfied,
then for all $f\in L^p(\Sg)$
\beq\label{91}
\begin{aligned}
(1-C_\om)^{-1} \, f &= m_+ \,(v^+)^{-1} = m_- (v^-)^{-1} \\
&= (M_+ + f)(v^+)^{-1} = (M_- +f) (v^-)^{-1}
\end{aligned}
\eeq
where $m_\pm$ solves $IRHP1_p$ with the given $f$  and $M_\pm$ solves
$IRHP2_p$ with $F=f(v-I)$ ($\in L^p$!), and
if $M_\pm$ solves $IRHP2_p$ with $F\in L^p (\Sg)$, then
\beqs
\begin{aligned}
M_+ &= \lf( (1-C_\om)^{-1} (C^-\, F)\rt) v^+ + F \quad\text{ and } \\
M_- &= \lf( (1-C_\om)^{-1} \, (C^- \,F)\rt) v^-.
\end{aligned}
\eeqs
\end{prop}
Finally, if $f\in L^\infty(\Sg)$ and $v^\pm - I  \in L^p(\Sg)$, then
\eqref{91} remains valid provided we interpret
\beq
(1-C_\om)^{-1}\, f \equiv f + (1-C_\om)^{-1} \;C_\om \, f.
\label{93}
\eeq
This is true, in particular, for the normalized RHP $(\Sg, v)_p$ where
$f \equiv I$.

\begin{remark}
If $1-C_\om$ is invertible, for one choice of $v^\pm$,  then (\tbf{exercise})
it is invertible for all choices of $v^\pm$ such that
\[
v= (v^-)^{-1} \, v^+, \qquad v^\pm, \quad (v^\pm)^{-1} \in L^\infty(\Sigma).
\]
\end{remark}

Note that if we take $v^+ = v, \quad v^-=I$, in particular, then
\beqs
C_\om\,h = C^- \lf(h(v-I)\rt).
\eeqs
The above Proposition implies, in particular, that if $\mu \in I + L^p$
solves
\beq
(1-C_\om)\,\mu =I
\label{95}
\eeq
in the sense of \eqref{93} i.e.\ $\mu=I+\nu$, $\; \nu\in L^p$
\beq
(1-C_\om)\, \nu = C_\om \, I = C^+\,\om_- + C^- \, \om_+
\in L^p
\label{96}
\eeq
then $m_\pm = \mu\,v^\pm$ solves the normalized RHP $(\Sg, v)_p$.
\tit{It is in this precise sense that the solution of the normalized RHP is equivalent
to the solution of a singular integral equation \eqref{95}, \eqref{96} on $\Sg$}.

One  very important consequence of the proof of Proposition \ref{prop1} is given by the following
\begin{corol}\label{cor1}
Let $f\in L^p(\Sg)$.

Let $m_\pm$ solve $IRHP1_p$ with the given $f$ and let $M_\pm$ solve $IRHP2_p$
with $F=f(v-I)$. Then
\begin{align}
\|(1-C_\om)^{-1}\, f\|_p &\le c\, \| m_\pm\|_p \label{97}\\
\intertext{and}
\| (1-C_\om)^{-1} \,f\|_p &\le c' \lf( \|M_\pm\|_p + \|f\|_p\rt)
\label{98}
\end{align}
for some constants $c= c_p,\; c'=c'_p$.  In particular if we know, or can
show, that $\|m_\pm \|_p \le \ndconst \|f\|_p$, or $\|M_\pm\|_p \le
\ndconst \|f\|_p$, then we can conclude from \eqref{97} or \eqref{98} that $(1-C_\om)^{-1}$
is bounded in $L^p$ with a corresponding bound.  Conversely if we know that
$(1-C_\om)^{-1}$ exists, then the above calculations show that $\|m_\pm\|_p
\le \tc \|f\|_p$ and $\|M_\pm \|_p \le \ttc \|f\|_p$ for corresponding constants
$\tc, \; \ttc$.
\end{corol}

Finally we consider uniqueness for the solution of the normalized RHP
$(\Sg, v)_p$ as given in Definition \ref{def3}.
Observe first that if $F(z) = (Cf)(z)$ for $f\in L^p(\Sg)$ and $G(z) =
(Cg)(z)$ for $g\in L^q(\Sg)$, $\; \frac{1}{r} = \frac{1}{p}+ \frac{1}{q}
\le 1$, $\; 1<p,\,q < \infty$, then a simple computation shows that
\beq
FG(z) = Ch(z)
\label{99}
\eeq
where
\beq
h(s) = - \frac{1}{2i}\; \lf(g(s) (Hf)(s) + f(s) (Hg)(s)\rt)
\label{100}
\eeq
where again $H\,f(s)= \text{ Hilbert transform } = \lim_{\ep \downarrow 0}
\int_{|s'-s| > \ep}\; \frac{f(s')}{s-s'}\; \frac{ds'}{\pi}$, and similarly
for $Hg(s)$.  As $h$ clearly lies in $L^r(\Sg),\; r \ge 1$, it follows
that
\beqs
F_+ \,G_+(z) - F_-\, G_-(z) = h(z) \qquad \text{for a.e. } \quad z\in \Sg.
\eeqs
(\tbf{Note:} $C^+\,h(z)- C^-\,h(z)=h(z)$ even if $h$ is in $L^1$, even though
$C^\pm$ is not bounded in $L^1$.)

\begin{theorem}\label{thm5}
Fix $1< p<\infty$.  Suppose $m_\pm$ solves the normalized RHP $(\Sg, v)_p$.
Suppose that $m^{-1}_\pm$ exists a.e.\ on $\Sg$ and $m^{-1}_\pm \in I + \p \,
C(L^q)$, $1<q < \infty$, $\frac{1}{r} = \frac{1}{p}+ \frac{1}{q} \le 1$.
Then the solution of the normalized RHP $(\Sg, v)_p$ is unique.
\end{theorem}
\begin{proof}
Suppose $\hm_\pm = I + C^\pm \, \hh, \;\;\hh\in L^p(\Sg)$ is a 2$^{nd}$
solution of the normalized RHP.  We have, by assumption, $m^{-1}_\pm =
I+C^\pm\,k$ for some $k\in L^q(\Sg)$.  (It is an \tbf{Exercise} to show
that $I+(Ck)(z)$, the extension of $m^{-1}_\pm$ to $\CC \backslash \Sg$,
is in fact $m(z)^{-1}$.).

Then arguing as above
\begin{align*}
\hm_\pm \, m^{-1}_\pm -I &= \lf(\hm_\pm -I\rt) \lf(m^{-1}_\pm - I\rt)
+ \lf(\hm_\pm -I\rt) + \lf(m^{-1}_\pm -I\rt)\\
&= C^\pm \, h
\end{align*}
for some $h\in L^r(\Sg) + L^p(\Sg) + L^q(\Sg)$.

\noi
Hence
\[
\hm_+ \;m^{-1}_+ - \hm_- \; m^{-1}_- = h.
\]
But
\[
\hm_+\, m^{-1}_+ = \lf(\hm_-\, v\rt) \lf(m_-\, v\rt)^{-1} = \hm_-\, m^{-1}_-
\]
and so $h=0$.  Thus $\hm_\pm\, m^{-1}_\pm-I=0$ or $\hm_\pm = m_\pm$.
\end{proof}

\begin{theorem}\label{thm6}
If $n=2,\; p=2$ and $\det v(z) =1 \text{ a.e. on } \;\Sg$, then the solution of
the normalized RHP $(\Sg, v)_2$ is unique.
\end{theorem}

\begin{proof}
Because $n=2$ and $p=2$, \eqref{99}, \eqref{100} $\Longrightarrow (\det m(z))_\pm = 1
+ C^\pm\,h$, where $h\in L^1(\Sg)+ L^2 (\Sg)$ and so $(\det m)_+ - (\det m)_-
= h(z)$ a.e.  But $\det m_+= \det m_-$ as $\det v=1$, and so $h\equiv 0$.
But then $\det m(z)_\pm =1$.  Hence, if
\begin{align*}
m_\pm &= \begin{pmatrix}
m_{11 \, \pm} & m_{12\, \pm} \\
m_{21 \, \pm} & m_{22\, \pm}\end{pmatrix} \\  
\intertext{we have}
m^{-1}_\pm &= \begin{pmatrix}
m_{22 \, \pm} &-m_{12\, \pm} \\
-m_{21 \, \pm} & m_{11\, \pm}\end{pmatrix}
\end{align*}
and so clearly $m^{-1}_\pm \in I + \p \,C(L^2)$.  The result now follows from Theorem \ref{thm5}.
\end{proof}

These results immediately imply that the normalized RHP $\lf(\Sg=\RR, \;
v_{x,t}\rt)$ for MKdV with $v_{x,t}$ given by \eqref{36}
has a unique solution in $L^2(\RR)$.
%
Indeed, factorize
\beqs
\begin{aligned}
v_{x,t}(z) = (v^-_{x,t})^{-1}\; \; v^+_{x,t} &= \lf(I-w^-_{x,t}\rt)^{-1}
\lf(I+ w_{x,t}^+\rt)
\\
&= \begin{pmatrix}
1 & -\bar{r}\,e^{-2i\tau} \\
0 & 1
\end{pmatrix} \;
\begin{pmatrix}
1 & 0 \\
re^{2i\tau} & 1
\end{pmatrix}
\end{aligned}
\eeqs
so that
\beqs
w_{x,t} = \lf(w^-_{x,t},\; w^+_{x,t}\rt) = \lf(
\begin{pmatrix}
0 & -\bar{r}e^{-2i\,\tau}\\
0&0
\end{pmatrix}, \;
\begin{pmatrix}
0 & 0 \\
re^{2i\,\tau} &0
\end{pmatrix}
\rt)
\eeqs
But for $\Sg=\RR$, we have
\begin{exer}\label{exer11}
Both $C^+$ and $-C^-$ are orthogonal projections in $L^2(\RR)$
and so $\|C^\pm \|_{L^2}=1$.
\end{exer}

Using the Hilbert-Schmidt matrix norm $\|M\|=
\lf(\sum_{i, j} \;|M_{ij}|^2 \rt)^{\frac{1}{2}}$, we have
\begin{align*}
\|C_{\om_{x,t}}\,h\|^2_{L^2} &= \lf\|C^+
\begin{pmatrix}
h_{11} & h_{12}\\
h_{21} & h_{22}
\end{pmatrix}
\begin{pmatrix}
0 & -\bar{r}\,e^{-2i\tau}\\
0&0
\end{pmatrix}\rt.
\\
&\qquad \qquad + C^- \lf(
\begin{pmatrix}
h_{11} & h_{12} \\
h_{21} & h_{22}
\end{pmatrix}
\begin{pmatrix}
0 & 0\\
re^{2i\tau} &0
\end{pmatrix}
\rt\|^2_{L^2} \\
&=\lf\| \begin{pmatrix}
0 & C^+\,h_{11} (-\bar{r})\, e^{2i\tau}\\
0&C^+\,h_{21} (-\bar{r})\, e^{-2i\tau}
\end{pmatrix} \rt. \\
&\qquad\qquad\qquad\lf. + \begin{pmatrix}
C^- \, h_{12}\,r e^{2i\tau} &0 \\
C^-\,h_{22} \, r e^{2i\tau} &0
\end{pmatrix} \rt\|^2
\\
&= \|C^-\,h_{12}\, r e^{2i\tau}\|^2_{L^2} + \|C^-\, h_{22} \,
re^{2i\tau}\|^2_{L^2} \\
& \quad + \|C^+\,h_{11} (-\bar{r})\,e^{-2i\tau}\|^2_{L^2} +
\|C^+\, h_{21} (-\bar{r})\, e^{-2i\tau} \|^2_{L^2} \\[3pt]
&\le \|r\|^2_\infty \lf( \|h_{12}\|^2_{L^2} + \|h_{22}\|^2_{L^2}
+\|h_{11}\|^2_{L^2} + \|h_{21}\|^2_{L^2} \rt)
\\[3pt]
&= \|r_\infty\|^2 \;\|h\|^2_{L^2}
\end{align*}
and so, as $\|r\|_\infty <1$,
\beqs
\|C_{\om_{x,t}}\|<1. 
\eeqs

It follows that for each $x, t\in \RR$, $\lf(1-C_{\om_{x,t}} \rt)^{-1}$
exists in $L^2(\RR)$ and
\beqs
\lf\| \lf( 1-C_{\om_{x,t}}\rt)^{-1} \rt\|_{L^2} \le \frac{1}{1-\|r\|_\infty}
< \infty 
\eeqs
and the proof of the existence and uniqueness for $\lf(\Sg, v_{x,t}\rt)$
follows from Proposition \ref{prop1}.  On the other hand, just uniqueness alone
follows from Theorem \ref{thm6} as $\det v(z) \equiv 1$ on $\RR$.

Now it turns out that a key role in the theory of RHPs is played by
\tit{Fredholm operators}.  Recall that a bounded linear operator $T$ from a
Banach space  $X$ to a Banach space $Y$ is \tit{Fredholm} if
\beqs
\dim \ker T < \infty
\eeqs
and
\beqs
\dim  \text{ coker }\; T < \infty
\qquad\text{i.e. $Y$/ran $T$ is a finite dimensional space.}
\eeqs
\beqs
\text{If $T$ is Fredholm,\quad we define index
$T \equiv \dim \ker T - \dim$ coker $T$.}
\eeqs
\begin{exer}\label{exer12}
If $T:X \to Y$ is Fredholm, then ran $T$ is closed in $Y$.
\end{exer}
\begin{exer}\label{exer13}
$T:X\to Y$ is Fredholm iff it has a \tit{pseudo-inverse}
$S\in \mL(Y, X)$ such that $ST = 1_X + K$ and $TS = 1_Y +L$
where $K$ is a compact operator in $\mL(X)$ and 
$L$ is a compact operator in $\mL(Y)$.
\end{exer}

We know that a normalized RHP $(\Sg, v)_p$, say, has a (unique) solution
if $(1-C_\om)^{-1}$ exists.  The situation where we know, for example,
that $\|C_\om\|_{L^2} <1$, as in the example $\lf(\Sg= \RR, v_{x,t}\rt)$
above so that $(1-C_\om)^{-1}$ exists,
is very rare.  For example, for the KdV equation on $\RR$
\begin{equation*}
	\begin{split}
		& u_t + 6 uu_x + u_{xxx} =0,\\
		& u(x,0) = u_0(x) \to 0 \qquad\text{as} \qquad |x|\to \infty,
	\end{split}
\end{equation*}
the associated RHP is exactly the same as $(\RR, v_{x,t})$ for MKdV, except
that now, generically,
\beq
|r(z)| <1 \qquad\text{for}\qquad |z| >0
\label{110}
\eeq
but
\beq
|r(0)| = 1.
\label{111}
\eeq
Thus
\(
\|r\|_\infty =1
\)
and the above proof of the existence and uniqueness for the RHP breaks down.
A more general approach to proving the existence and uniqueness of
solutions to normalized RHPs, is to attempt the following:
\def\ind{\mathrm{ind}}
\begin{itemize}
\item Prove $1-C_\om$ is Fredholm.
\item  Prove $\ind \lf(1-C_\om \rt)  =0$.
\item Prove $ \dim \ker \lf(1-C_\om\rt)  =0 $.
\end{itemize}
Then it follows that $1-C_\om$ is a bijection, and hence the normalized RHP
$(\Sg, v)$ has a unique solution.

Let's see how this goes for KdV with normalized RHP $\lf(\Sg=\RR,
v_{x,t}\rt)$,  but now $r$ satisfies \eqref{110}, \eqref{111}.
By our previous comments (see \tbf{Remark} above), it is enough to consider the
special case $v^+=v, \;\; v^-=I\;\;$ so that $\om^+=v-I\;$ and $\;\om^-=0$.
Thus
\beqs
C_\om\,h= C^-\,h\lf(v-I\rt).
\eeqs
We assume $r(z)$ is continuous and $r(z) \to 0$ as $|z| \to \infty$.
Let $S$ be the operator
\beqs
Sh= C^-\,h\lf(v^{-1}-I\rt).
\eeqs
Then
\begin{align*}
C_\om\, Sh &= C^- \lf(Sh(v-I)\rt) \\
&= C^- \lf[ \lf( C^- h \lf(v^{-1}-I\rt) \rt)(v-I)\rt]\\
&= C^- \lf[\lf( C^+ h \lf(v^{-1}-I\rt) \rt)(v-I) \rt] \\
&\quad -C^- \lf[ h  \lf(v^{-1}-I\rt) (v-I) \rt]
\end{align*}
as $C^+ - C^- =1$.  But $h\lf(v^{-1}-I\rt) (v-I)=h\lf(2I-v-v^{-1}\rt)
= h(I-v)+h\lf(I-v^{-1}\rt)$.  

Thus
\begin{align*}
C_\om \, Sh &= C^- \lf[ \lf(
C^+\, h \lf(v^{-1}-I\rt) \rt)(v-I) \rt] \\
& + C^- \lf(h \lf(v^{-1}-I\rt)\rt) + C^- \lf(h \lf(v-I\rt)\rt) \\
&= C^- \lf[ \lf(C^+ \,h\lf(v^{-1}-I\rt) \rt) (v-I)\rt]+ C_\om\,h + Sh
\end{align*}
and we see that
\[
(1-C_\om) (1-S)h =h +C^- \lf[ \lf( C^+\,h \lf(v^{-1}-I\rt) \rt) (v-I)\rt].
\]
But
\begin{exer}\label{exer14}
$\quad K\,h= C^- \lf[ \lf(C^+\,h \lf(v^{-1}-I\rt) \rt) (v-I) \rt] \quad
\text{is compact in $L^2(\RR).$}$\
\end{exer}
Hint: $v-I$ is a continuous function which $\to 0$ as $|z| \to \infty$ and
hence can be approximated in $L^\infty(\RR)$ by finite linear combinations of
functions of the form $a/(z-z')$ for suitable constants $a$ and points $z' \in
\CC \backslash \RR$.
Then use the following fact:
\begin{exer}
If $T_n, \; n \ge 1$ are compact operators in $\mL (X, Y)$ and $\|T_n-T\|
\to 0$ as $n\to \infty$ for some operator $T\in \mL (X,Y)$, then $T$
is compact.
\end{exer}

Similarly
\[
(1-S) (1-C_\om) =1 + L, \qquad L \text{ compact.}
\]
Thus $(1-C_\om)$ is Fredholm.

Now we use the following fact:
\begin{exer}
Suppose that for $\ga \in [0,1]$, $T(\ga)$ is a norm-continuous family of
Fredholm operators.  Then for $\ga \in [0,1]$,
\beqs
\ind T(\ga) = \ndconst = \ind T(0) = \ind T(1).
\eeqs
\end{exer}
Apply this fact to $C_{\om(\ga)}$, where we replace $r$ by $\ga r$ in $v_{x,t}$,
\begin{align*}
v_{x,t,\ga}  = \begin{pmatrix}
1- \ga^2|r|^2 &-\ga\,\bar{r} e^{-2i\tau} \\
\ga \,r e^{2i\tau} &1
\end{pmatrix}.\end{align*}
The proof above shows that $C_{\om(\ga)}$ is a norm continuous
family of Fredholm operators and so $\ind (1-C_\om) = \ind
\lf(1-C_{\om(\ga=1)} \rt) = \ind \lf(1-C_{\om(\ga=0)}\rt)=0$ as
$C_{\om(\ga=0)}=0$ and the index of the identity operator is clearly $0$.

Finally suppose
\[
(1-C_\om) \, \mu=0.
\]
Then using \eqref{90}, $m_+ = \mu v$ and $ m_-=\mu$ solve
$m_+ = m_-\,v, m_\pm \in \p C(L^2)$.

Consider $P(z) = m(z) \lf(m(\bar{z})\rt)^*$ for $z\in \CC^+$ where $m(z)$ is
the extension of $m_\pm\;$ off $\;\RR$ i.e. if $m_\pm = C^\pm\,h$,
$\;h\in L^2$, then $m(z)= (Ch)(z)$.  Then for a contour $\Ga_{R,\, \ep}$, pictured in Figure~\ref{fig:19},
$\int_{\Ga_{R, \ep}} \,P(z)\,dz=0$ as $P(z)$ is analytic.

\begin{figure}[H]
\centering
\begin{tikzpicture}[scale=.7]
\draw[line width = 1] (-2,0) to (2,0);
\node [below] at (0,0) {$0$};
\draw[line width = 1, directed] (-1.95, .4) to (1.98,.4);
\draw[->] (0,0) to (1.2,1.6);
\node [left] at (.8, 1.05) {\scalebox{.8}{$R$}};

\node[]  at(-1,.19) {$\epsilon$};
\node[above]  at(-1,.24) {$\downarrow$};
\node[below]  at(-1,0.15) {$\uparrow$};
\draw[line width = 1, directed, domain = 12:168] plot ({2*cos(\x)}, {2*sin(\x)});
\end{tikzpicture}
\caption{\label{fig:19} A semi-circle $\epsilon$ above the real axis.} 
\end{figure}
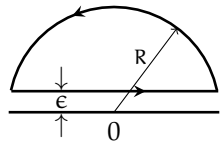

Letting $\ep \downarrow 0$ and $R\to\infty$, we obtain (\tbf{exercise})
$\int^\infty_{-\infty} P_+(z)\, dz=0$; i.e.
\[
0= \int_\RR m_+(z) \; m_-(z)^*\,dz = \int_\RR
m_-(z) \, v(z) \, m_-(z)^*\, dz.
\]
Taking adjoints and adding, we find
\[
0 =  \int_\RR m_-(z) \lf(v+v^*\rt)(z) \; m_-(z)^*\, dz.
\]
But a direct calculation shows that $v+ v^*$ is diagonal and
\[
\lf(v+v^*\rt)(z) = 2 \begin{pmatrix}
1-|r(z)|^2 & 0 \\
0 & 1
\end{pmatrix} \ .
\]
Now since $|r(z)|<1$ a.e. (in fact everywhere except $z=0$), we conclude that
$m_-(z)=0$. But $\mu =m_-$ and so we see that $\ker \lf(1-C_\om\rt)= \{0\}$.

The result of the above chain of arguments is that the solution of the normalized RHP
$\lf(\Sg, v_{x, t}\rt)$ for KdV exists and is unique.  Such Fredholm
arguments have wide applicability in Riemann--Hilbert Theory \cite{FokasPainleve}.

One last general remark.  The scalar case $n=1$ is special.  This is because
the RHP can be solved explicitly by formula.  Indeed, if $m_+=m_-\,v$,
then it follows that $(\log m)_+ = (\log m_-) + \log v$ and hence $\log m(z)$ is given by
Plemelj's formula, which provides the general solution of \tit{additive} RHPs,
via
\[
\log m = C\lf(\log v\rt)(z) = \int_\Sg \frac{\log v(s)}{s-z}\; \dts
\]
and so
\beq
m(z) = \exp\lf( \int_\Sg \;\frac{\log v(s)}{s-z}\; \dts\rt)
\label{116}
\eeq
a formula which is easily checked directly.  However, there is a hidden
subtlety in the business:  On  $\RR$, say, although $v(s)$ may go rapidly
to $0$ as $s\to \pm \infty$, $v(s)$ may wind around $0$ and so $\log v(s)$
may not be integrable at both $\pm \infty$.  Thus there is a \tit{topological}
obstacle to the existence of a solution of the RHP.  If $n>1$, there are many
more such ``hidden'' obstacles.

\section*{Lecture 4}
\setcounter{equation}{0}\setcounter{section}{4}\addcontentsline{toc}{section}{Lecture 4}

RHP's arise in many difference ways.
For example, consider orthogonal polynomials: we are given a measure $\mu$
on $\RR$ with finite moments,
\beqs
\int_\RR |x|^m \, d\mu(x) < \infty \qquad \text{ for } \quad m=0, 1, 2, \dots
\eeqs
Performing Gram-Schmidt on $1, x, x^2, \dots$ with respect to $d\mu(x)$,
we obtain (monic) orthogonal polynomials
\beqs
\pi_n(x) = x^n + \dots, \qquad , \;n\ge 0
\eeqs
such that
\beqs
\int_\RR \pi_n(x) \, \pi_m(x) \, d\mu u(x) =0, \qquad n\neq m, \qquad n, m \ge 0.
\eeqs
(Here we assume that $d\mu$ has infinite support: otherwise there are only a finite
number of such polynomials.)

Associated with the $\pi_n$'s are the orthonormal polynomials
\beq
P_n(x) = \ga_n \,\pi_n(x), \qquad \ga_n>0, \qquad n \ge 0
\label{120}
\eeq
such that
\beqs
\int_{\RR} P_n(x) \, P_m(x) \, d\mu(x) = \de_{n,m}, \qquad n, m\ge 0.
\eeqs
Orthogonal polynomials are of great historical and continuing importance in
many different areas of mathematics, from algebra, through combinatorics, to
analysis.  The classical orthogonal polynomials, such as the Hermite
polynomials, the Legendre polynomials, the Krawchouk polynomials, are well
known and much is known about their properties.  In view of our earlier
comments it should come as no surprise that much of this knowledge,
particularly  asymptotic properties, follows from the fact that these
polynomials have integral representations analogous to the integral
representation for the Airy function in the first lecture.  For example,
for the Hermite polynomials
\[
\int_{\RR} H_n(x) \, H_m(x) \, e^{-x^2} \, dx =0 \qquad n \neq m, \qquad n, m \ge 0
\]
one has the integral representation
\beqs
H_n(x) = n! \int_{\msC} \om^{-n-1} \, e^{2x\om-\om^2}\, d\om
\eeqs
where $\msC$ is a (small) circle enclosing the origin,
(Note: the $H_n$'s are not monic, but are proportional to the $\pi_n$'s,
$H_n(x) = c_n\, \pi_n(x)$ where the $c_n$'s are explicit)
and the asymptotic behavior of the $H_n$'s follow from the classical
steepest descent method.  For general weights, however, no such integral
representations are known.  

The Hermite polynomials play a key role in
random matrix theory in the so-called Gaussian Unitary, Orthogonal and
Symplectic Ensembles.  However it was long surmised that local properties
of random matrix ensembles were \tit{universal},  i.e., independent of the
underlying weights.  In other words if one considers general weights such as
$e^{-x^4}\, dx, \; e^{-(x^6+ x^4)}\, dx,$ etc., instead of the weight
$e^{-x^2}\,dx$
for the Hermite polynomials, the local properties of the
random matrices, at the technical level, boil down to analyzing the asymptotics
of the polynomials orthogonal with respect to the weights $e^{-x^4}\, dx$,
$e^{-(x^6+ x^4)}\, dx$, etc., for which no integral representations are known.
What to do?

It turns out however, that orthogonal polynomials with respect to
an \tit{arbitrary} weight can be expressed in terms of a RHP.
Suppose $d\mu(x)= \om(x)\,dx$, for some $\om(x)\ge 0$ such that
\beqs
\int_\RR |x|^m \, \om (x)\, dx < \infty, \qquad m=0,1,2, \dots.
\eeqs
and suppose for simplicity that
\beq
\om \in H^1(\RR) = \{ f \in L^2: f' \in L^2 \}.
\label{124}
\eeq

Fix $n \ge 0 $ and let $Y^{(n)}= \{ Y^{(n)}_{ij} (z) \}_{1 \le i, \, j \le 2}$
solve the RHP $\lf(\Sg= \RR, v=\lf( \begin{smallmatrix}
1 &\om \\
0 & 1
\end{smallmatrix}\rt) \rt)$ normalized so that
\beqs
Y^{(n)}(z)
\begin{pmatrix}
z^{-n} &0 \\
0 & z^n
\end{pmatrix} \to I \qquad \text{as }\qquad z\to \infty.
\eeqs
\begin{exer}\label{exer15} Show that we then have (see e.g.\ \cite{DeiftOrthogonalPolynomials})
\beqs
Y^{(n)}(z) = \begin{pmatrix}
\pi_n(z) & C(\pi_n\, \om) \\
-2\pi i\, \ga_{n-1}^2 \, \pi_{n-1} (z) & C\lf(-2\pi\, i\ga^2_{n-1}\;
\pi_{n-1}\, \om \rt)
\end{pmatrix}
\eeqs
where $C= C_\RR$ is the Cauchy operator on $\RR$, $\pi_n, \;\pi_{n-1}$ are the
monic orthogonal polynomials with respect to $\om(x)\, dx$ and
$\ga_{n-1}$ is the normalization coefficient for $\pi_{n-1}$ as in \eqref{120}.
(Note that by \eqref{124} and Theorem \ref{thm2}, $Y^{(n)} (z)$ is continuous down to the
axis for all $z$.) This discovery is due to Fokas, Its and Kitaev \cite{FIK}.
Moreover this is just exactly the kind of problem to which
the nonlinear steepest descent method can be applied to obtain
(\cite{Deift1999,DeiftWeights4}) the
asymptotics of the $\pi_n$'s with comparable precision to the classical cases,
Hermite, Legendre, $\dots$, and so prove universality for unitary ensembles
(and later, Deift and Gioev, Shcherbina, for Orthogonal \& Symplectic Ensembles
of random matrices, see \cite{DeiftGioev} and the references therein).

As mentioned earlier, RHPs are useful not only for asymptotic analysis, but
also to analyze analytical and algebraic issues.  Here we show how RHPs give
rise to difference equations, or differential equations, in other situations.
\end{exer}

Consider the solution $Y^{(n)}$ for the orthogonal polynomial RHP $
\lf(\RR, \,v=\lf( \begin{smallmatrix} 1 &\om \\ 0 & 1
\end{smallmatrix}\rt) \rt)$.  The \tit{key fact} is
that the jump matrix $\lf(\begin{smallmatrix} 1 &\om \\ 0 & 1
\end{smallmatrix}\rt)$ is
\tit{independent} of $n$: the dependence on $n$ is only in the boundary
condition
\[
Y^{(n)} \begin{pmatrix}
z^{-n} &0 \\
0 & z^{+n}
\end{pmatrix} \to I \ .
\]
So we have $Y^{(n+1)}_+ = Y^{(n+1)}_- \, v$ and $Y^{(n)}_+ = Y^{(n)}_-\, v$.

Let $ R(z) = Y^{(n+1)}(z) \lf(Y^{(n)} (z) \rt)^{-1},  z\in \CC
\backslash \RR$.
Then
\begin{align*}
R_+(z) &= Y^{(n+1)}_+(z) \lf( Y^{(n)}_+ (z) \rt)^{-1}\\
&= \lf( Y^{(n+1)}_- (z) \; v(z) \rt) \lf( Y^{(n)}_- (z) \; v(z) \rt)^{-1} \\
&= Y^{(n+1)}_-(z) \lf(v(z) \; v(z)^{-1} \rt) \lf( Y^{(n)}_- (z) \rt)^{-1}\\
&= R_-(z).
\end{align*}
Hence $R(z)$ has no jump across $\RR$ and so, by an application of Morera's
Theorem, $R(z)$ is in fact entire.  But as $z\to \infty$
\begin{align*}
R(z) &= \lf[ Y^{(n+1)}(z) \begin{pmatrix}
z^{-n-1} &0 \\
0 & z^{n+1}
\end{pmatrix} \rt]
\begin{pmatrix}
z& 0 \\
0 & z^{-1}
\end{pmatrix} \lf[ Y^{(n)} (z)
\begin{pmatrix}
z^{-n} &0 \\
0 & z^n
\end{pmatrix} \rt]^{-1}\\
&= \lf( I +O \lf(\frac{1}{z} \rt) \rt)
\begin{pmatrix}
z & 0 \\
0 & z^{-1}
\end{pmatrix} \lf( I + O\lf(\frac{1}{z}\rt) \rt)
\\
&= O(z).
\end{align*}
Thus $R(z)$ must be a polynomial of order 1,
\begin{equation*}
Y^{(n+1)}(z) \lf(Y^{(n)}(z) \rt)^{-1} = R(z) = Az + B
\end{equation*}
for suitable $A$ and $B$, or,
\begin{equation}
Y^{(n+1)}(z) = \lf(Az + B\rt) Y^{(n)}(z) \label{127}
\end{equation}
which is  a difference equation for orthogonal polynomials with respect to a
fixed weight.
\begin{exer}\label{exer16}
Make the argument leading to \eqref{127} rigorous (why does
$\lf(Y^{(n)}\rt)^{-1}$ exist, etc.)
\end{exer}
\begin{exer}\label{exer17}
Show that \eqref{127} implies the familiar three term recurrence relation for
orthogonal polynomials $p_n(z)$
\[
b_n\;p_{n+1}(z) + \lf(a_n-z\rt) p_n(z) + b_{n-1}\, p_{n-1} =0, \qquad n \ge 0
\]
$a_n \in \RR, \quad b_n >0; \quad b_{-1} \equiv 0$.
\end{exer}

Whereas the RHP for orthogonal polynomials comes ``out of the blue'',
there are some systematic methods to produce RHP representations for
certain problems of interest.  This is true in particular for RHPs
associated with ordinary differential equations.  For example,
consider the ZS--AKNS equation (Zakharov-Shabat, Ablowitz-Kaup-Newell-Segur)
\beq
\lf[\p_x - \lf(i z\sg + \begin{pmatrix}
0 & q(x) \\
\bar{q}(x) &0
\end{pmatrix} \rt) \rt] \psi =0,\qquad -\infty < x < \infty
\label{128}
\eeq
(see e.g.\ \cite{deiftzhounls}).
Here $z\in \CC, \quad \sg =  \half \lf( \begin{smallmatrix}
1 & 0 \\
0 &-1
\end{smallmatrix}\rt)\;$ and $q(x) \to 0 $ at some sufficiently fast rate as
$|x| \to \infty $.  Equation \eqref{128} is intimately connected with the
defocusing Nonlinear Schr\"{o}dinger Equation (NLS) by virtue of the fact
that the operator
\beq
L= (i \sg)^{-1} \lf( \p_x - \begin{pmatrix}
0 & q \\
\bar{q} &0
\end{pmatrix} \rt)
\label{129}
\eeq
undergoes an \tit{isospectral deformation} if $q=q(t) = q(x,t)$ solves NLS
\beq\label{130}
\begin{aligned}
iq_t + q_{xx} -2|q|^2\, q=0 \\
q(x, t=0) = q_0(x).
\end{aligned}
\eeq
In other words, if $q=q(t)$ solves NLS then the \tit{spectrum} of
\[L(t)= (i\sg)^{-1} \lf( \p_x - \lf(\begin{smallmatrix}
0 & q(x,t) \\
\ovl{q(x,t)} &0
\end{smallmatrix} \rt)
\rt)\] is \tit{constant:} Thus the spectrum of $L(t)$ provides constants
of the motion for \eqref{130}, and so NLS is ``integrable''.  The key fact is
that there is a RHP naturally associated with $L$ which expresses the
integrability  of NLS in a form that is useful for analysis.  Here we
follow Beals and Coifman, see \cite{BC}.  Let $q(x)$ in \eqref{129} be given with $q(x)\to 0$ as $|x|\to
\infty$ sufficiently rapidly.  Then for any $z\in \CC \backslash \RR$,
\begin{exer}\label{exer18}
The equation $(L-z)\,\psi=0$ has a \tit{unique} solution $\psi$ such that
$\psi \,(x,z)\, e^{-ixz \sg}\ \to I$ as $x\to -\infty$ and is bounded $x \to \infty$. Such 
$\psi\,(x,z)$ are called \tit{Beals-Coifman} solutions.
\end{exer}

\begin{rem}\label{rmkprops}
	These solutions have the following properties:
\begin{enumerate}
\item For fixed $x$, $\psi(x,z)$ is analytic in 
$\CC \backslash \RR$,
and is continuous down to the axis. That is $\psi_\pm (x,z) = \lim_{\ep \downarrow 0} \psi\lf(x, z\pm i\,\ep\rt)$
exist for all $x,z\in \RR$.
\item For fixed $x$,
$\psi(x,z) e^{-ixz\sg}\to I$ as $ z\to\infty$,
\beq
\psi (x,z)\, e^{-ixz\sg} = I + \frac{m_1(x)}{z} + O \lf(\frac{1}{z^2}\rt),
\qquad\text{as}\qquad z\to \infty
\label{135}
\eeq
for some matrix residue term $m_1(x)$.
\end{enumerate}
\end{rem} 
Now clearly $\psi_\pm (x,z),\;\;z\in \RR$, are two fundamental solutions of
$(L-z)\,\psi=0$ and so for $z\in \RR$,
\beqs
\psi_+\lf(x,z\rt) = \psi_- \lf(x,z\rt) \; v(z)
\eeqs
for all $x\in \RR$, where $v(z)$ is \tit{independent} of $x$.
In other words, by (1) of Remark~\ref{rmkprops}, $\psi(x, \cdot)$ \tit{solves a RHP $\lf(\Sg=
\RR,\,v\rt)$, normalized as in \eqref{135}}.  In this way differential equations
give rise to RHPs in a systematic way.

One can calculate (\tbf{exercise}) the precise form of $v(z)$ and one finds
\beqs
v(z) = \begin{pmatrix}
1-|r(z)|^2 & r(z) \\
-\ovl{r(z)} &1
\end{pmatrix}, \qquad z\in \RR
\eeqs
where, again (cf. \eqref{36} for MKdV) we have for $r$, the \tit{reflection
coefficient},
\beqs
\|r\|_\infty < 1.
\eeqs
Now the map
\beq \label{eq:reflect}
q \mapsto  r = \mR(q)
\eeq
is a bijection between suitable spaces: $r=\mR(q)$, the direct map, is
constructed from $q$ via the solutions $\psi(x,z)$ as above.  The inverse map
$r\mapsto \mR^{-1}(r) =q$ is constructed by solving the RHP $(\Sg, v)$
normalized by \eqref{135} for any fixed $x$.  One obtains
\begin{align}
\psi(x,z) \, e^{-izx\sg} &= I + \frac{m_1(x; r)}{z} + O\lf(\frac{1}{z^2}\rt)
\qquad\text{as} \qquad z\to \infty
\notag\\
\intertext{and}
q(x) &= -i \lf(m_1(x, r)\rt)_{12} \notag
\end{align}
(cf \eqref{39} for MKdV).

Now if $q=q(t)= q(x,t)$ solves NLS then $r(t) = \mR\lf(q(t)\rt)$ evolves simply,
\beqs
r(t) = r(t,z) = r(t=0, z)\, e^{-itz^2}\ , \qquad z\in \RR
\eeqs
i.e.  $t\to q(t) \to r(t) \to \log \, r(t) = \log r(t=0)
-itz^2$
\tit{linearizes} NLS.  This leads to the following formula for the solution of
NLS with initial data $q_0$
\beq
q(t) = \mR^{-1} \lf(e^{-it(\cdot)^2}\; \mR(q_0)(\cdot)\rt).
\label{142}
\eeq
The effectiveness of this representation, which one should view as the RHP
analog of NLS of the integral representation \eqref{10} for the Airy equation,
depends on the effectiveness of the nonlinear steepest descent method
for RHPs.

\begin{ques}
Where in the representation \eqref{142} is the information encoded that $q(t)$
solves NLS?
\end{ques}

The \tit{answer} is as follows.  Let $\psi(x,z, t)$ be the solution of the
RHP with jump matrix
\[
v_t (z)= \begin{pmatrix}
1-|r|^2 &r e^{-itz^2}\\
-\bar{r} e^{itz^2} &1
\end{pmatrix}
\]
normalized as in \eqref{135}.  Set $H(x,z,t)=\psi(x,z,t)\, e^{-itz^2\,\sg}$
and observe that
\beq
H_+ = H_- \begin{pmatrix}
1- |r|^2 &r \\
-\bar{r} &1
\end{pmatrix} = H_-\,v
\label{143}
\eeq
for which the jump matrix is \tit{independent of $x$ and $t$.  This means
that we can differentiate \eqref{143} with respect to $x$ and $t$},\;
$H_{x+} = H_{x-}\,v,\;H_{t+}= H_{t-}\,v$ and conclude, as in the case
of orthogonal polynomials, that $H_x\, H^{-1}$ and $H_t\, H^{-1}$ are
entire, and evaluating these combinations as $z\to \infty$, we obtain two
equations
\beqs
H_x = D\,H \qquad, \quad H_t = E\,H
\eeqs
for suitable polynomials matrix functions $D$ and $E$.  These functions
constitute the famous Lax pair $(D,E)$ for NLS. Compatibility
of these two equations requires
\begin{align*}
&\p_t \,\p_x\,H = \p_x \, \p_t\, H\\
\Longrightarrow\qquad  & \p_t (D\,H) = \p_x \, (E\,H)\\
\Longrightarrow\qquad  & D_t\,H+ D\, E\,H = E_x\,H + E\,D\,H\\
\Longrightarrow\qquad  & D_t+ \lf[D,\, E\rt] = E_x
\end{align*}
which reduces directly to NLS.
In this way RHP's lead to difference and differential equations.

Another systematic way that RHP's arise is through the distinguished class
of so-called \tit{integrable operators}.  Let $\Sg$ be an oriented contour
in $\CC$ and let $f_1, \dots, f_n$ and $g_1, \dots, g_n$ be bounded measurable
functions on $\Sg$.  We say that an operator $K$ acting on $L^p(\Sg),\;\;
1<p<\infty$, is \tit{integrable} if it has a kernel of the form
\beqs
K(z,z') = \frac{ \Sg^n_{i=1} \;f_i(z) \;g_i(z')}{z-z'},
\qquad z,\, z' \in \Sg;\qquad z\neq z'
\eeqs
for such $L^\infty$ functions $f_i,\, g_j$,
\[
(K\,h)(z) = \int_\Sg K(z, z')\, h(z')\, dz'\ .
\]
Integrable operators were first singled out as a distinguished class of
operators by Sakhnovich \cite{sakh} in the late 1960's, and their theory was
developed fully by Its, Izergin, Korepin and Slavnov \cite{IIKS} in the early 1990's
(see \cite{Deift1999a} for a full discussion).
The famous sine kernel of random matrix theory
\beq \notag
K_x(z,z') = \frac{\sin x (z-z')}{\pi(z-z')} =
\frac{e^{ixz} \, e^{-ixz'}
+ \lf(-e^{ixz'}\rt)\cdot
e^{ixz}}{2i\,\pi(z-z')}
\eeq
is a prime example of such an operator, as is likewise the well-known
Airy kernel operator.

Integrable operators form an algebra, but their
\tit{most remarkable property} is that their
\tit{inverses can be expressed} in terms of the solution
of a naturally associated RHP.  Indeed, let $m(z)$ be the solution of the
normalized RHP $(\Sg, v)$ where
\begin{equation}\label{147}
	v(z)  = I -2\pi i f\,g^T,\quad f= \lf(f_1, \dots, f_n\rt)^T, g=\lf(g_1, \dots, g_n\rt)^T.
\end{equation}
(Here we assume for simplicity  that $\Sg^n_{i=1}\,f_i(z)\, g_i(z)=0$,
for all $z\in \Sg$ as in the sine-kernel: otherwise \eqref{147} must be
slightly modified).

Then $(1-K)^{-1}$ has the form $1+L$ where $L$ is an integrable operator
\beqs
L(z, z') = \frac{\Sg^n_{i=1}\;F_i(z)\;G_i(z')}{z-z'}, \qquad
z, z' \in \Sg, \qquad z\neq z'
\eeqs
and
\beq\label{149}
\lf\{
\begin{aligned} 
F &= \lf(F_1, \dots, F_n\rt)^T = m_\pm\, f \\
G &= \lf(G_1, \dots, G_n\rt)^T = (m_\pm^{-1})^T\, g.
\end{aligned}
\rt.
\eeq
This means that if, for example, $K$ depends on parameters, as in the
case of the sine kernel, asymptotic problems involving $K$ as the parameters
become large, are converted into asymptotic problems for a RHP, to which
the nonlinear steepest descent method can be applied.

As an example, we show how to use RHP methods to give a proof of Szeg\H{o}'s
celebrated Strong Limit Theorem.  Let $\TT$ be the unit circle.

\begin{theorem}[Szeg\H{o} Strong Limit Theorem]\label{thm7}
Let $\varphi(z) =e^{L(z)} \in L^1(\TT),\; \varphi (z)>0$, where
$\sum^\infty_{k=1} k|L_k|^2 < \infty$ and $\{L_k\}$ are Fourier coefficients
of $L(z)$.  Let $D_n$ be the Toeplitz determinant generated by $\varphi$,
$D_n (\varphi) = \det \; X(\varphi)$ where $X(\varphi)$ is the $(n+1)
\times (n+1)$ matrix with entries
$\{ \varphi_{i-j}\}_{0 \le i,\; j \le n}$,
and $\{\varphi_k\}$ are the Fourier coefficients of $\varphi$.  Then as
$n\to\infty$,
\beqs
D_n= e^{\displaystyle(n+1)L_0+ \Sg^\infty_{k=1} k|L_k|^2} \lf(1+o(1)\rt).
\eeqs
\end{theorem}
\begin{proof}[Sketch of proof] Let $e_k, \;\;0\le k \le n $, be the standard basis
in $\CC^{n+1}$.  Then the map $U_n:e_k \to z^k$, $\;0\le k\le n$,
$\;z\in \TT$ takes $\CC^{n+1}$ onto the trigonometric polynomials $\mP_n
= \lf\{ \sum^n_{j=0}\,a_j\, z^j\rt\}$
of degree $n$ and induces a map
\[
\tau_n : \mP_n \to \mP_n
\]
which is conjugate to $X(\varphi)$.

We then calculate 
\begin{equation}\label{151}
\begin{split}
	\tau_n \, z^k
	&= U_n \;X\;U^{-1}_n\; z^k  \\
	&= U_n \;X\;e_k \\
	&= U_n \Bigl(\sum^n_{j=0} \varphi_{j-k}\, e_j\Bigr)  \\
	&= \sum^n_{j=0} \varphi_{j-k}\;z^j, \qquad 0 \le k \le n. 
\end{split}
\end{equation}
Now for any $p= \sum^n_{k=0} a_k\, z^k \in \mP_n$
\begin{align}
\lf(\tau_n\,p\rt)(z) &= \sum^n_{k=0} a_k \sum^n_{j=0} \varphi_{j-k}\,z^j \notag\\
&= \sum^n_{k=0} a_k \sum^n_{j=0} \lf( \int_{\Gamma} (z')^{k-j-1} \,
\varphi(z')\, \dtz'\rt) z^j \notag\\
&= \sum^n_{k=0} a_k \int_\Gamma (z')^{k-1}\, \varphi(z') \;
\frac{(z/z')^{n+1}-1}{(z/z')-1}\; \dtz' \notag\\
&= \int_\Ga \varphi(z')\, p(z')\; \frac{(z/z')^{n+1}-1}{(z-z')} \;\dtz'.
\notag
\end{align}
After some simple calculations (\tbf{Exercise}) one finds that
\beq
\tau_n \, p = \lf(1-K_n\rt)p, \qquad p\in \mP_n
\label{153}
\eeq
where $K_n$ is the integrable operator on $\TT$ with kernel of the form
\beq
K_n \lf(z,\, z'\rt) = \frac{f_1(z)\, g_1(z') + f_2(z) \, g_2(z')}{z-z'}
, \qquad z, \,z'\in \Ga
\label{154}
\eeq
where
\begin{equation}\label{155}
	\begin{split}
		f& = \lf(f_1, f_2\rt)^T = \lf(z^{n+1}, 1\rt)^T\\\
		g &= \lf(g_1, g_2\rt)^T = \lf(z^{-n-1}\;\frac{1-\varphi(z)}{2\pi i},
		- \frac{\lf(1-\varphi(z)\rt)}{2\pi i} \;\rt)^T.
	\end{split}
\end{equation}

We have, in particular, from \eqref{151} and \eqref{153}, for $0 \le k \le n$,
\beqs
\lf(1-K_n\rt)z^k = \sum^n_{j=0} \varphi_{j-k}\, z^j
\eeqs
and for $k<0$ and $k>n$ one easily shows that
\beqs
\lf(1-K_n\rt)z^k= z^k + \sum^n_{j=0} \varphi_{j-k}\, z^i.
\eeqs
Thus $K_n$ is finite rank, and hence trace class, and $\lf(1-K_n\rt)$ has
block form with respect to the orthonormal basis $\{z^k\}^\infty_{-\infty}$
for $L^2(\Ga)$ as given in Figure~\ref{fig:block}. And so
\beqs
D_n = \det \,\tau_n = \det X(\varphi) = \det \lf(1-K_n\rt)
\eeqs

\begin{figure}[H]
\centering
\begin{tikzpicture}[scale=0.55]
\node at (-1.5,1.5) {\scalebox{1.2}{$I$}};
\node at (0,1.5) {\scalebox{1.2}{$0$}};
\node at (1.5,1.5) {\scalebox{1.2}{$0$}};
\node at (-1.5,0) {\scalebox{1.2}{$\cdots$}};
\node at (0,0) {\scalebox{1.2}{$\tau_n$}};
\node at (1.5,0) {\scalebox{1.2}{$\cdots$}};
\node at (-1.5,-1.5) {\scalebox{1.2}{$0$}};
\node at (0,-1.5) {\scalebox{1.2}{$0$}};
\node at (1.5,-1.5) {\scalebox{1.2}{$I$}};

\draw[line width=1] (-2, .75) to (2,.75);
\draw[line width=1] (-2, -.75) to (2,-.75);
\draw[line width=1] (.75, -2) to (.75, 2);
\draw[line width=1] (-.75, -2) to (-.75, 2);
\end{tikzpicture}
\caption{\label{fig:block} The block structure of $1-K_n$ in the basis $\{z^k\}^\infty_{-\infty}$.  }
\end{figure}

Associated with the integrable operator $K_n$ we have the normalized RHP
$\lf(\Sg= \Ga, v\rt)$ where, by \eqref{147}, \eqref{155}
\begin{equation}\label{159}
v= I -2\pi i \, f g^T
 =\begin{pmatrix}
\varphi & - (\varphi-1)\, z^{n+1}\\
z^{-n-1} \, (\varphi-1) & 2-\varphi
\end{pmatrix}
\end{equation}
on $\TT$. Now
\begin{equation}\label{160}
\begin{split}
	\log \, D_n &= \log \, \det\lf(1-K_n\rt)  \\
	&= tr \; \log \lf(1-K_n\rt) \\
	&= \int^1_0 \frac{d}{dt}\; \text{ tr } \log \lf(1-t\, K_n\rt) dt \\
	&=- \int^1_0 \text{ tr }\lf( \frac{1}{1-t\,K_n}\;K_n \rt) dt.
\end{split}
\end{equation}
For $0\le t\le 1$, set
\beqs
\varphi_t(z) = (1-t) + t\, \varphi(z), \qquad z\in \mathbb T.
\eeqs
Clearly $\;\;\varphi_t (z)>0\;\;$ and $\;\;\varphi_0(z) =1,\;\;
\varphi_1(z) = \varphi(z)$.
Now $\;\;\varphi_t-1=t\lf(\varphi-1\rt)\;\;$ and so we have from \eqref{154}
\beqs
t\,K_n = K_{t, \, n} = \lf[ \lf( (z/z')^{n+1} - 1\rt)/ (z-z') \rt]
\lf[ \lf(1-\varphi_t(z')\rt)/ 2\pi i\rt]
\eeqs
and it follows that in \eqref{160}
\begin{align*}
\frac{1}{1-t\,K_n} \; t\,K_n &= \frac{1}{1-K_{t,\,n}} \; K_{t,\,n} \\
&= \frac{1}{1-K_{t,n}} - 1\\
&= R_{t,n}
\end{align*}
where
\beqs
R_{t,\,n} \lf(z, z'\rt) =
\frac{\sum^2_{j=1} F_{t, j}(z)\,G_{t,j}(z')}{z-z'}
\eeqs
where by \eqref{149}
\beq\label{164}
\lf\{
\begin{aligned}
F_t &= \lf(F_{t, 1},\, F_{t,2} \rt)^T = m_{t\,\pm}\,f_t, \\
G_t &= \lf(G_{t,1},\, G_{t,2}\rt)^T = \lf(m_{t \,\pm}^{-1} \rt)^T\, g_t.
\end{aligned}
\rt.
\eeq
Here $m_{t\,\pm}$ refers to the solution of the RHP $(\TT,\,v_t)$ where $v_t$
involves $\varphi_t$ rather than $\varphi$ in \eqref{159},
and similarly for $f_t,\, g_t$.

Hence (\tbf{Exercise})
\beq
\log\, D_n = -\int^1_0 \lf( \int_\TT
\lf(\sum^2_{j=1}\, F'_{t,j}(z) \,
G_{t,j}(z) \rt)
dz\rt) \frac{dt}{t}\ .
\label{165}
\eeq
So we see that in order to evaluate $D_n$ as $n\to\infty$
we must evaluate the asymptotics of the solution $m_t$ of the normalized
RHP $\lf(\TT, \,v_t\rt)$ as $n\to\infty$, for each $0\le t \le 1$, and
substitute this information into \eqref{165} using \eqref{164}. This is precisely
what can be accomplished \cite{Deift1999a} using the nonlinear steepest descent method.

Here we present the nonlinear steepest descent analysis in the case when $\varphi(z)$ is analytic in an annulus
\[
A_\ep = \{ z: 1-\ep < |z| < 1+ \ep\}, \qquad \ep >0
\]
around $\TT$.  The idea of the proof, which is a common feature of all
applications of the nonlinear steepest descent method, is to move  the
$z^{n+1}$ term (or its analog in the general situation) in $v_t$ into
$|z|<1$ and the $z^{-n-1}$ term into $|z|>1$: then as $n\to \infty$,
these terms are exponentially small, and can be neglected.

But first we must separate the $z^{n+1}$ and $z^{-n-1}$ terms of
$v_t$ algebraically.  This is done
using the lower-upper pointwise factorization of $v_t$
\beq \label{166}
v_t = \begin{pmatrix}
1 & 0 \\
z^{-n-1} \lf( 1-\varphi^{-1}_t\rt) &1
\end{pmatrix}
\begin{pmatrix}
\varphi_t &0 \\
0 & \varphi^{-1}_t
\end{pmatrix}
\begin{pmatrix}
1 & -\lf(1- \varphi^{-1}_t \rt) z^{n+1}\\
0 & 1
\end{pmatrix}
\eeq
which is easily verified.

Extend $\TT =\Sg \to \tilde{\Sg} = \lf\{ |z|=\rho\}
\cup \Sg \cup
\{|z|= \rho^{-1}\rt\} = \Sg_\rho \cup \Sg\cup \Sg_{\rho-1}$
where we choose $1-\ep < \rho <1 < \rho^{-1} < 1+ \ep$.  Now define a piecewise
analytic function $\tilde{m}$ by the definitions in Figure~\ref{fig:piece}.

\begin{figure}[H]
\centering
\begin{tikzpicture}[scale=0.78]
\node at (0,0) {$\tilde{m} = m$};
\draw[line width=1, directed, domain = 0:360, samples = 200] plot ({1*cos(\x)}, {1*sin(\x)});
\draw[line width=1, directed, domain = 0:360, samples = 200] plot ({2*cos(\x)}, {2*sin(\x)});
\draw[line width=1, directed, domain = 0:360, samples = 200] plot ({3*cos(\x)}, {3*sin(\x)});
\node[right] at (1,0){$\Sigma_p$};
\node[right] at (2,0){$\Sigma$};
\node[right] at (3,0){$\Sigma_{p-1}$};
\node at (3,3) {$\tilde{m} = m$};

\draw[->] (2.5,-2.5) to (1.2, -1.2);
\node[right] at (2.5,-2.7){\scalebox{.75}{$\tilde{m} = m\left(\begin{array}{cc}
1 & -(1-\phi_t) z^{n+1} \\
0 & 1
\end{array}\right)^{-1}$}};

\draw[->] (-3,3) to (-1.8, 1.8);
\node[left] at (-3,3){\scalebox{.75}{$\tilde{m} = m\left(\begin{array}{cc}
1 & 0 \\
z^{-n-1}(1-\phi_t^{-1}) & 1
\end{array}\right)$}};
\end{tikzpicture}
\caption{\label{fig:piece} A piecewise definition of $\tilde m$. }
\end{figure}
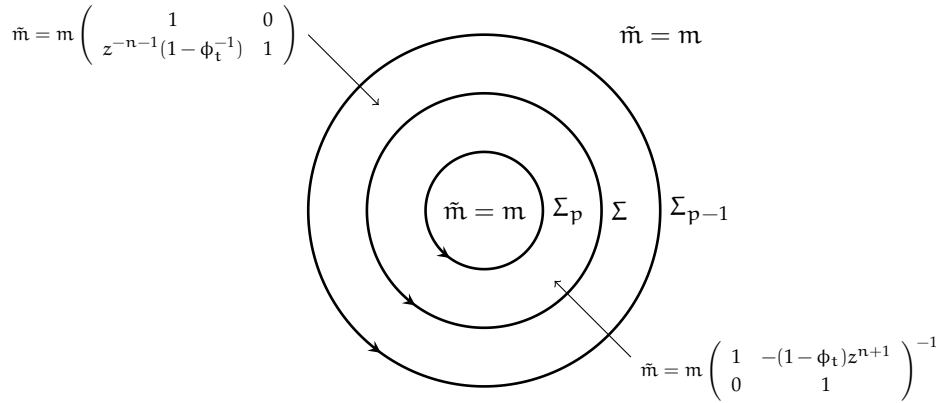

This definition is motivated by the fact that
\[
m_+ = m_- v_t = m_- ( \cdot ) ( \cdot ) ( \cdot )
\]
as in \eqref{166}. It follows that $\tilde{m}(z)$ solves the normalized RHP $\lf(\tilde{\Sg},
\tilde{v}\rt)$ where
\begin{alignat}{2}
\notag\tilde{v}(z) &= \begin{pmatrix} 1 & 0 \\
z^{-n-1} \lf(1-\varphi^{-1}_t\rt) &1
\end{pmatrix} &\qquad\text{on }\qquad \Sg_{\rho^{-1}},\\
\notag\tilde{v}(z) &= \begin{pmatrix} \varphi_t(z) & 0 \\
0 & \varphi_t(z)^{-1}
\end{pmatrix} &\qquad\text{on }\qquad \Sg,\\
\notag\tilde{v}(z) &= \begin{pmatrix} 1 &  - \lf(1-\varphi^{-1}_t\rt) z^{n+1} \\
0 &1
\end{pmatrix} &\qquad\text{on }\qquad \Sg_\rho\ .
\end{alignat}
Now as $n\to\infty$, $\;\tilde{v}(z) \to I$ on $\Sg_\rho$ and on
$\Sg_{\rho-1}$.  This means that $\tilde{m}\to m_\infty$ where $m_\infty$
solves the normalized RHP $\lf(\Sg, v_\infty\rt)$ where
\beqs
v_\infty = v\big|_\Sg = \begin{pmatrix}
\varphi_t &0 \\
0 & \varphi^{-1}_t
\end{pmatrix}.
\eeqs
But this RHP is a direct sum of \tit{scalar} RHP's and hence can be solved
explicitly, as noted earlier (cf. \eqref{116}).  In this way we obtain the
asymptotics of $m$ as $n\to\infty$ and hence the asymptotics of the Toeplitz
determinant $D_n$.
\end{proof}
Here is what, alas, I have not done and what I had hoped to do in these lectures
(see AMS open notes):
\begin{itemize}
\item Show that in addition to the usefulness of RHP's for algebraic
and asymptotic purposes, RHP's are also
useful for analytic purposes.  In particular, RHP's can be used to show
that the Painlev\'e equations indeed have the Painlev\'e property.
\item Show that in addition to RHP's arising ``out of the blue'' as in
the case of orthogonal polynomials and systematically in the case of
ODE's and also integrable operators, RHP's also arise in a systematic
fashion in Wiener--Hopf Theory.
\item Describe what happens to an RHP when the operator $1-C_\om$ is
Fredholm, but not bijective, and
\item Finally, I have not succeeded in showing you how the nonlinear
steepest descent method works in general.  All I have shown is one simple
case.
\end{itemize}

%
%
%
%
%
%
%

\bibspread


\end{document}